\documentclass[a4paper,twocolumn,11pt,accepted=2022-04-16]{quantumarticle}
\pdfoutput=1

\usepackage[margin=0.8in]{geometry}
\usepackage[utf8]{inputenc}
\usepackage[english]{babel}
\usepackage[T1]{fontenc}
\usepackage{amsmath,amssymb,amsthm}
\usepackage{mathtools}

\usepackage{subcaption}
\usepackage{tikz}
\usetikzlibrary{decorations.markings,arrows,patterns,shapes}

\usepackage[font=small,labelfont=bf]{caption}
\usepackage{float}
\usepackage{url}

\usepackage{mathrsfs}
\usepackage{dsfont}
\usepackage{setspace}
\usepackage[export]{adjustbox}
\usepackage{makecell}
\usepackage{tqft}
\usepackage{enumitem}

\usepackage{xcolor}
\definecolor{darkblue}{RGB}{0,0,128}
\definecolor{darkgreen}{RGB}{0,150,0}

\usepackage[numbers,sort&compress]{natbib}
\usepackage[pdfusetitle]{hyperref}
\hypersetup{breaklinks, colorlinks, linkcolor=blue, citecolor=darkgreen, filecolor=red, urlcolor=darkblue}

\newtheorem{theorem}{Theorem}
\newtheorem*{theorem*}{Theorem}
\newtheorem{lemma}[theorem]{Lemma}

\newtheorem{corollary}[theorem]{Corollary}

\newtheorem{definition}[theorem]{Definition}

\usepackage{thmtools}
\usepackage{thm-restate}

\RequirePackage{filecontents}

\usepackage{mystyle}
\graphicspath{{../}{./}}

\usepackage{authblk}
\begin{document}

\title{Connectivity constrains quantum codes}
\author[1]{Nouédyn Baspin}
\thanks{nouedyn.baspin@usherbrooke.ca}
\affiliation{Universit\'e de Sherbrooke, Sherbrooke, Qu\'ebec, Canada J1K 2R1}
\author[2]{Anirudh Krishna}
\thanks{anirudhk@stanford.edu}
\affiliation{Stanford University, Stanford, CA, USA, 94305}

\date{}

\maketitle

\begin{abstract}
	Quantum error correcting codes are a scheme through which a set of measurements is used to correct for decoherence in a quantum system. 
	Due to experimental limitations, it is natural to require that each of these measurements only involve a constant number of qubits.
    This requirement motivates the class of quantum low-density parity-check (LDPC) codes, which also limits the number of measurement outcomes a qubit can affect.
	Seminal results have shown that quantum LDPC codes implemented through local interactions in $D$-dimensional Euclidean space obey strong restrictions on their code dimension $k$, distance $d$, and their ability to implement fault-tolerant operations. 
	However, we lack an understanding of what limits quantum LDPC codes that do not have an explicit embedding in $\bbR^D$.
	The need for a more general understanding of these limitations is highlighted by recent breakthroughs in the construction of LDPC codes that eschew locality, and yet witness tradeoffs between code parameters.
	In this work we prove bounds applicable to any quantum LDPC code.
	
    Our main results are a) a bound on the distance, b) a bound on the code dimension and c) limitations on certain fault-tolerant gates that can be applied to quantum LDPC codes.
    All three of these bounds are cast as a function of the graph separator of the connectivity graph representation of the quantum code.
    We find that unless the connectivity graph contains an expander, the code is notably limited.
    This implies a necessary, but not sufficient, condition to construct good codes.
    This is the first bound that studies the limitations of quantum LDPC codes that does not rely on geometric locality.
    As an application, we present bounds on quantum LDPC codes associated with local graphs in $D$-dimensional hyperbolic space, and local graph on $g$-genus surfaces.
\end{abstract}

\begin{center}
    \textbf{Dedicated to the memory of David Poulin}
\end{center}

\thispagestyle{empty}
\setcounter{page}{0}
\vfill

\pagebreak

\section{Introduction}
A fault-tolerant quantum circuit will require error correction at regular intervals to avoid the build up of errors \cite{aharonov1997fault, aliferis2005quantum, kitaev1997quantum, knill1998resilient, shor1996fault}.
The error correcting code used is assessed using various figures-of-merit.
Of these, the two most fundamental are the code dimension $k$ and the distance $d$.
The code dimension $k$ is the number of qubits that can be encoded in the code.
The distance $d$ measures the number of single-qubit errors required to irreparably corrupt encoded information.
The choice of code also affects how encoded information can be processed in a quantum circuit.
We want to design a code in a way that protects encoded information from unavoidable interactions with the environment which might corrupt the code; yet at the same time, we want the code to be amenable to interactions that facilitate computation.
Understanding the optimal tradeoff between these three figures-of-merit is a fundamental question in quantum error correction \cite{gottesman1996class,ashikhmin1999upper,sarvepalli2010degenerate,knill1997theory,eldar2020need}.
In this paper, we study these tradeoffs in the context of quantum low-density parity-check (LDPC) codes.

A quantum LDPC code is characterized by how syndrome information is gathered.
Unlike the classical setting, we cannot directly read quantum codewords.
The state of a register of $n$ qubits could be in some delicate superposition which measurements can upset.
The only information we can use for diagnosis is the syndrome, itself just a binary string.
Each bit of the syndrome is obtained by measuring a set of qubits in a manner prescribed by the error correcting code.
These specific measurements are designed to preserve the encoded information and are called \emph{stabilizer measurements}.
Each bit of the syndrome is obtained by measuring the corresponding \emph{stabilizer generator}.
Together, the stabilizer generators generate a stabilizer \emph{group}, a set of measurements that does not destroy encoded quantum information.
In a quantum LDPC code, we need only measure a constant number of qubits for each bit of the syndrome.
Furthermore, each qubit only affects the value of at most a constant number of syndrome bits.
This property is expected to simplify the process of obtaining the syndrome which, in addition to the code dimension and distance, is also a criterion for picking a quantum error correcting code.
Indeed, quantum LDPC codes may have benefits for constructing scalable fault-tolerant quantum circuits \cite{kovalev2013fault, gottesman2014fault, fawzi2018constant}.
In sharp contrast to the classical setting, it is unknown whether good quantum LDPC codes exist i.e.\ whether quantum LDPC code families exist where $k$ and $d$ scale linearly with $n$.

For ease of implementation, we may wish to construct quantum LDPC codes that are spatially local in $2$ dimensions.
A \emph{local} quantum code refers to a code family embedded in $\bbR^D$ in which the qubits involved in a particular syndrome bit are contained in a ball of diameter $w$, where $w$ is some constant, independent of the size of the code.
Unfortunately, locality is a fundamental problem in the design of quantum error correcting codes.
Bravyi and Terhal \cite{bravyi2009no} proved that any local code in $\bbR^D$ obeys $d = O(n^{1-1/D})$.
This bounds the distance of a $D$-dimensional local code away from $n$.
Subsequently, Bravyi, Poulin and Terhal \cite{bravyi2010tradeoffs} proved that any local code in $\bbR^D$ obeys $kd^{2/(D-1)} = O(n)$.
In particular, $2$-dimensional codes are very restricted: their distance $d$ can scale at best as $\Theta(\sqrt{n})$, implying that the code dimension is constant.
The famous surface code (and the closely related color code and variants) saturates this bound up to constant factors \cite{kitaev2003fault, bravyi1998quantum, bombin2006topological, bombin2010topological}.

Constructive approaches that eschew locality to build quantum LDPC codes still face difficulties.
There exist codes that achieve a code dimension scaling linearly in the block size but with limited distance \cite{tillich2014quantum, bravyi2014homological, kovalev2012improved, zeng2019higher,freedman2002z2,guth2014quantum,londe2017golden,leverrier2020quantum}.
It proved to be very challenging to achieve a distance scaling better much better than $\Theta(\sqrt{n})$ \cite{freedman2002z2}.
In the latter half of 2020, a series of works heralded one breakthrough after another \cite{kaufman2014ramanujan, evra2020decodable,kaufman2020quantum,hastings2020fiber,breuckmann2020balanced}.
The current record is held by a construction due to Panteleev and Kalachev, who demonstrated the existence of codes with code dimension $k = \Theta(\log(n))$ and distance scaling as $\Theta(n/\log(n))$ \cite{panteleev2020quantum}.

In contrast to these constructive approaches, we present Bravyi-Poulin-Terhal-like bounds applicable to general LDPC codes that are not constrained to be local.
Such a top-down approach to bound the properties of quantum LDPC codes might serve to answer why finding constructive approaches has been difficult.

In addition to these concerns, we need ways to process encoded information fault tolerantly.
This means that if a subroutine within a circuit fails, it only corrupts the limited set of qubits it acts on.
We do not want errors in one location to spread to errors in another, thereby overwhelming the error correcting code.
Transversal gates are one way to implement a fault-tolerant gate \cite{nielsen2002quantum}.
In its simplest form, transversal gates refer to gates acting independently on each physical qubit in the code.
Note that this is trivial in the classical setting: to implement the logical $\sansserif{NOT}$ on a $3$-bit repetition code, we need just flip each bit of the code.
However, this is considerably more difficult in the quantum setting as the set of transformations even on just a single qubit corresponds to $SU(2)$, a dense group.
Bravyi and Koenig \cite{bravyi2013classification} proved that transversal gates on $D$-dimensional local quantum error correcting codes are limited.
Specifically, transversal gates on $2$-dimensional local codes can at best implement transformations of a finite group of transformations referred to as the Clifford group.
This finite group of transformations is insufficient to implement all gates required to run interesting algorithms and can even be simulated efficiently on a classical computer \cite{aaronson2004improved}.
Subsequently, Pastawski and Yoshida showed that there is a relation between the distance of $D$-dimensional local codes and the gates they support \cite{pastawski2015fault}.
To state their result in a non-technical way, they proved that implementing transformations outside the finite group would come at the cost of the distance of the code.
Recent work by Burton and Browne \cite{burton2020limitations} extends this result and has shown that a specific class of finite rate quantum LDPC codes (that are not constrained by locality) has a structure where transversal gates can still only implement Clifford transformations on encoded information.
This suggests that locality itself might not be the constraint that limits transversal gates.

\subsection{Summary of results}
As previously mentioned it is known that locality is a strong limitation to the protection of quantum information. 
The purpose of the present work is to show that \emph{connectivity}, in a well defined sense, imposes similar restrictions on quantum codes.
To illustrate the difference, consider a code on $n$ qubits that is embedded in a a tree lattice. We require that we are limited to local interactions: if two qubits are involved in the same measurement, then they necessarily are neighbors in the lattice.

We can then make two seemingly conflicting observations. First, a tree lattice is highly non-local as it cannot be embedded in any $D$-dimensional space, hence it is not restricted by the Bravyi, Terhal and Poulin bounds. On the other hand, a tree lattice is poorly connected: if we remove the root vertex, then large chunks of the code can no longer communicate. Intuitively this should restrict the properties of the code, and we show that this is indeed the case.

This example is meant to underline the difference between locality and connectivity. Any notion of locality has to refers to a particular space -- Euclidean, hyperbolic, etc. -- while connectivity is meant to only refer to the inner structure of the code. The main metric of connectivity we will use here is the \emph{separator}

The separator $\sep(G)$ of a graph $G$ is a subset of vertices of minimal size which, if removed, would split $G$ into two small subgraphs that are disconnected from each other \cite{teng1991phd}.
Colloquially, if every subgraph $H$ of $G$ has a small separator, then $G$ is poorly connected: in some sense the graph is shallow.
This notion has recently received some interest in the field of geometric group theory under the name of \emph{separation profile} \cite{benjamini2012separation,hume2017continuum,coz2020separation,hume2020poincare}, which we write $s_G$.
The function $s_G$ bounds the size of the separator not only for the graph $G$, but also for all of its subgraphs.

For our results it will be crucial we require that the connectivity be low \emph{everywhere}, which is more appropriately captured by $s_G$ than by the size of a single separator. 
This is especially important when $G$ is, for example, made of several disconnected but dense graphs.
In this case, we would have $\sep(G) = \emptyset$, but this hardly captures the geometry of the entire graph. 

\textbf{Example:} Consider the $\sqrt{n} \times \sqrt{n}$ grid graph.
This graph has a separator of size $\sqrt{n}$: we just have to remove a single column of vertices from the middle to cleave the graph in two.
In fact, any planar graph is poorly connected; the famous Lipton-Tarjan Theorem states that any planar graph with $n$ vertices has a separator of size $O(\sqrt{n})$ \cite{lipton1979separator}. 

Finally, we note that the separator allows us to bound the dimension, and the performance of transversal gates in a code.
We detail these results in sections \ref{subsec:distancebnd}, \ref{subsec:dimensionbnd} and \ref{subsec:transversal}.
These results build on the idea of partitioning the qubits into subsets that do not contain a logical operator.
As shown in \cite{bravyi2009no,bravyi2010tradeoffs,bravyi2013classification}, the size and number of such subsets can provide a lot of information about the tradeoff between $k$ and $d$, as well as the ability of transversal gates to induce logical transformations.
We use the separation profile to construct such partitions by recursively separating the code into smaller parts.

We have not yet discussed to what degree these bounds are practical. 
Can they easily be applied to a given class of codes?\footnote{The separator itself is hard to compute. Surprisingly, there exist polynomial-time algorithms to \emph{approximate} the separator of an arbitrary graph up to constant factors \cite{orecchia2012lorenzo}.}
Fortunately, numerous separator theorems are known \cite{kawarabayashi2010separator,kisfaludi2020hyperbolic, dujmovic2015genus, gladkova2020separation, coz2020separation, lipton1979separator}.
These theorems, for a given class of graphs, guarantee upper bounds on their separators.
With these tools, we can address an open question of \cite{breuckmann2021ldpc}, where it is asked whether bounds on the parameters of local codes in non-Euclidean lattices can be obtained. 
In section \ref{subsec:hyperbolic}, we answer this question in the affirmative by proving bounds on local codes in $D$-dimensional hyperbolic space $\bbH^D$. 
These bounds follow naturally from a recent result by Kisfaludi-Bak \cite{kisfaludi2020hyperbolic} who showed that graphs locally embedded in $D$-dimensional hyperbolic space $\bbH^{D}$ have bounded separators.

In comparison with the known bounds for local codes in $D$-dimensional Euclidean space, the bounds we obtain for $\bbH^D$ are more restrictive on the distance but offer the same tradeoff between $k$ and $d$.
Indeed we find $kd^{2/(D-1)} = O(n)$ for a local code in $\bbH^D$, the same as for $\bbR^D$ as shown by Bravyi, Poulin \& Terhal.
If we find codes to saturate the bounds in $D$-dimensional Euclidean space, we would conclude that $\bbH^D$ does not seem advantageous. 
We note, however, that these bounds do not apply to hyperbolic manifolds, which have been used to prove the existence of constant rate codes with polynomial distance \cite{breuckmann2016constructions,londe2017golden,guth2014quantum,hastings2013decoding}.

Similarly, we use a result of Dujmovi\'c, Eppstein and Wood \cite{dujmovic2015genus} to prove bounds on codes locally embedded on a surface of genus $g$, thus extending a result of Delfosse \cite{delfosse2013tradeoffs} to arbitrary LDPC codes.

Our main results are presented in Section \ref{sec:main}.
The first of these results, Theorem \ref{thm:distancebnd}, states that the distance is bounded by $s_G$.
Secondly, Theorem \ref{thm:dimbndcmax} shows that a small separation profile implies a stark tradeoff between $k$ and $d$.
Similarly, a small separation profile implies a limited ability to perform transversal gates as shown in Theorem \ref{thm:transpoly}.
We note that Theorems \ref{thm:dimbndcmax} and \ref{thm:transpoly} are not limited to LDPC codes.

The rest of the paper proves these results and explores their consequences.
In Section \ref{sec:background}, we establish background required to state our result.
We define quantum stabilizer codes in Section \ref{subsec:stab}, some important properties and their representations in terms of graphs.
In Section \ref{subsec:sep-treewidth}, we define the notions of separability and the closely related notion of treewidth.
These are metrics of connectivity in terms of which our main theorems are stated.
In Section \ref{sec:main} we formally state and prove our main results.
First Section \ref{subsec:distancebnd} focuses on Theorem \ref{thm:distancebnd} on the distance.
Then Section \ref{subsec:dimensionbnd} focuses on Theorem \ref{thm:dimbndcmax} on the code dimension.
Lastly, Section \ref{subsec:transversal} focuses on Theorem \ref{thm:transpoly} on transversal gates.

\section{Background and Notation}
\label{sec:background}
\subsection{Stabilizer codes}
\label{subsec:stab}
In our paper, we focus on stabilizer codes on $n$ qubits.
A qubit is associated with the complex Euclidean space $\bbC^{2}$ and $n$ qubits with $(\bbC^{2})^{\otimes n}$.
Let $\cP$ denote the $n$-qubit Pauli group and for any two operators $\ssP, \ssQ \in \cP$, let $[\ssP, \ssQ] = \ssP\ssQ - \ssQ\ssP$ denote their commutator.
An $\dsl n,k \dsr$ quantum code $\cC$ is a $2^k$-dimensional subspace of the $n$-qubit space $(\bbC^{2})^{\otimes n}$.
It is specified by the stabilizer group $\cS$, an Abelian subgroup of the $n$-qubit Pauli group that does not contain $-\ssI$.
The code space $\cC$ is the set of states left invariant under the action of the stabilizer group, i.e.\ $ \cC= \{ \ket{\psi} : \; \ssS \ket{\psi} = \ket{\psi} \forall \ssS \in \cS\}$.
$\cS$ is generated by $(n-k)$ independent generators.
We may consider codes with an over-complete set of stabilizers generators of size $m \geq n-k$.
Technically not all of these elements are independent but we shall call them generators for convenience when no confusion may arise.

For any code $\cC$ discussed here, we implicitly assume assume that its stabilizer group $\cS$ is generated by a set $\Omega$ of commuting Pauli operators. More formally, we have $\cS = \langle \Omega \rangle$. When we want to underline this dependence on $\Omega$, we refer to $\cC$ as $C(\Omega)$.

Let $\cL = \{\ssL : [\ssL, \ssS] = 0 , \forall \ssS \in \cS \}$ denote the logical operators: the group of Pauli operators that commute with $\cS$ and preserve $\cC$.
The action of the logical operators on the code space $\cC$ can be classified by the quotient group $\cL/\cS = \{[\ssL]: \ssL \in \cL \}$ where we let $[\ssL]$ denote the coset $\ssL \cS = \{ \ssL \ssS : \ssS \in \cS\}$.
If $\ssL,\ssL'$ belong to the same coset then their action on $\cC$ is equivalent : $\ssL \ssL'  \in \cS$.
We write $\ssL \sim \ssL'$ when they belong to the same class. 

The code dimension $k$ corresponds to the number of nontrivial, independent elements of $\cL/\cS$.
The distance $d = \min_{\ssL \in \cL \setminus \cS} |\supp(\ssL)|$ is a metric to estimate the closeness of codewords.

To implement a universal set of gates, we require a finite set of gates that can approximate any unitary on $n$ qubits to the desired level of precision.
Typically, this set is chosen to be the set of gates in the Clifford group together with one additional gate that is not in the Clifford group.
The Clifford group is a finite group and corresponds to the automorphism group of the Pauli group.

The Clifford hierarchy is a generalization of the Clifford group and plays an important role in the theory of quantum error correction \cite{gottesman1999demonstrating}.
The $\ell$-th level of the Clifford hierarchy, denoted $\cK^{(\ell)}$, is defined recursively.
The $1$st level of the Clifford hierarchy, denoted $\cK^{(1)}$, is the Pauli group.
For $\ell \geq 2$, the hierarchy is defined as
\begin{align}
	\cK^{(\ell)} = \{ W : W \ssP W\conj \in \cK^{(\ell-1)} \quad \forall \ssP \in \cP_n \}~.
\end{align}
It can be seen from this definition that $\cK^{(2)}$ is the Clifford group.

Let $W$ be a unitary gate and let $\overline{W}$ denote the encoded version of $W$.
In other words, if $\cE_{\cS}: \bbC^{\otimes k} \embeds \bbC^{\otimes n}$ is the encoding operation for the quantum error correcting code defined by the stabilizer group $\cS$, then $\overline{W}\cE_{\cS}(\ket{\phi}) = \cE_{\cS}(W\ket{\phi})$.
We say that $W$ can be implemented in a transversal manner if $\overline{W} = W_1 \otimes ... \otimes W_n$ for some single-qubit gates $\{W_{i}\}_i$.

We recall some definitions from \cite{bravyi2009no}.
Let $V = [n]$ index the qubits and $Q(V)$, the $n$-qubit space associated with $\left(\bbC^{2}\right)^{\otimes n}$.
For any subset $U \subseteq V$, let $Q(U) \subseteq Q(V)$ denote the $|U|$-qubit space with the corresponding indices.
For ease of notation, we shall use $U$ to also refer to $Q(U)$ where no confusion may arise.
Let $\comp{U} = [n] \setminus U$ be the complement of $U$. For any Pauli operator $\ssL$ we write its support $\supp(\ssL) \subset V$ the set of qubits on which $\ssL$ acts nontrivially.
Central to everything that follows is the notion of \emph{correctability} of sets of qubits. 

\begin{definition}[Correctable set]
	For $U \subset V$, $U$ is \emph{correctable} if there exists a recovery map $\cR: \comp{U} \to V$ such that for any code state $\rho_{\cC} \in \cC$, $\cR(\Tr_{U}(\rho_{\cC})) = \rho_{\cC}$.
\end{definition}

Acting on the region $U$ alone cannot alter the information contained in the entire code in a meaningful way.
This idea is formalized by Lemma \ref{lem:cleaning} known as the Cleaning lemma \cite{bravyi2009no}:

\begin{lemma}[Cleaning Lemma]
	\label{lem:cleaning}
	Suppose the code $\cC(\Omega)$ has at least one nontrivial logical operator.
	For any subset $U \subset V$,
	\begin{enumerate}
		\item there is a non-trivial $\ssL \in \cL$ that is supported entirely in $U$, or
		\item for all $[\ssL] \in \cL/\cS$, there is a representative in $\ssL' \in [\ssL]$ such that $\ssL'$ acts trivially on $U$.
		One has $\ssL' \ssL = \prod_i \ssS_i$, with $\{\ssS_i\} \subseteq \Omega$ a set of generators  and the support of each $\ssS_i$ overlaps with $U$.
	\end{enumerate}
\end{lemma}

As a consequence, correctable regions turn out to have a rather interesting property.
If the support of a logical operator $\ssL$ intersects a correctable region $U$, then it can be \emph{cleaned} out of $U$. 
This is illustrated in fig. \ref{fig:clean}.

\begin{figure}[h]
	\centering
	\includegraphics[width=\columnwidth]{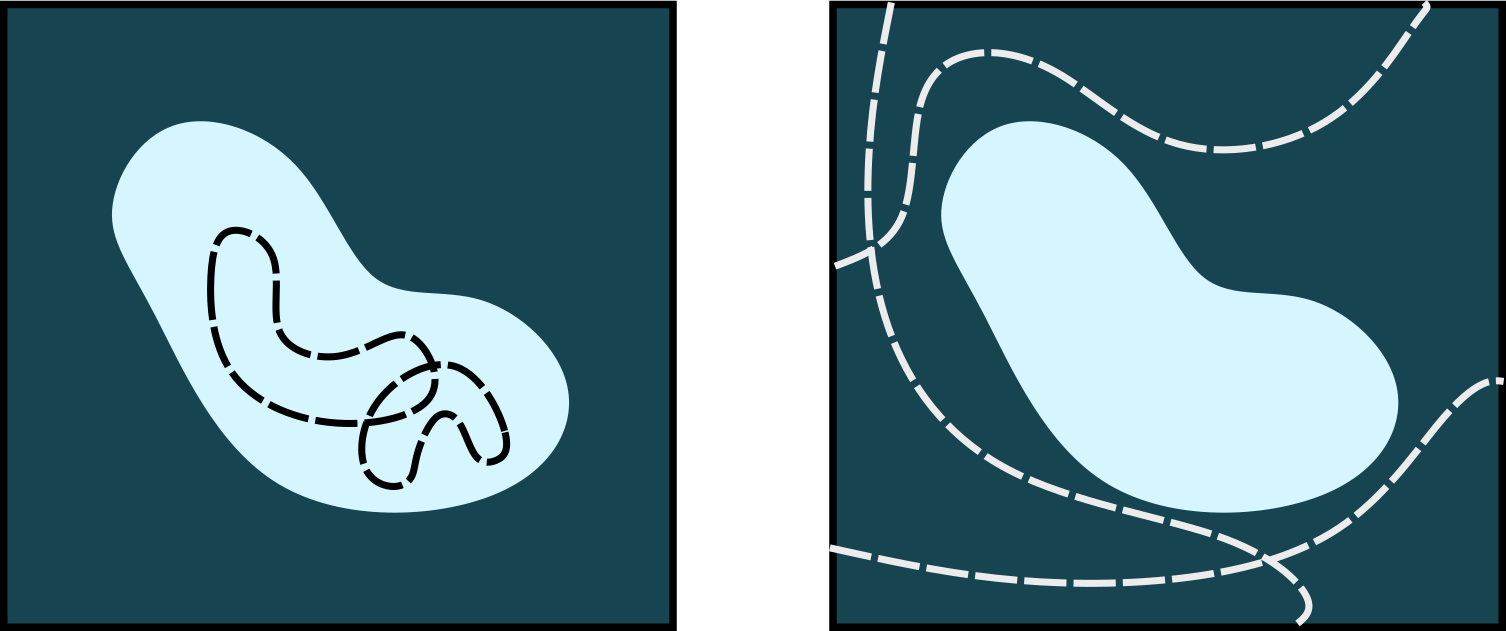}
	\caption{\textbf{Visualizing the Cleaning Lemma:} In this schematic, $V$ is the dark region of all qubits and $U$ is the light region within.
		The support of logical operators is depicted using a dashed line.
		Either the region $U$ contains logical operators as in (a) or \emph{all} logical operators can be made to run outside it as in (b).}
	\label{fig:clean}
\end{figure}

Intuitively we expect that subsets of qubits that are sufficiently far away from one another can be corrected independently.

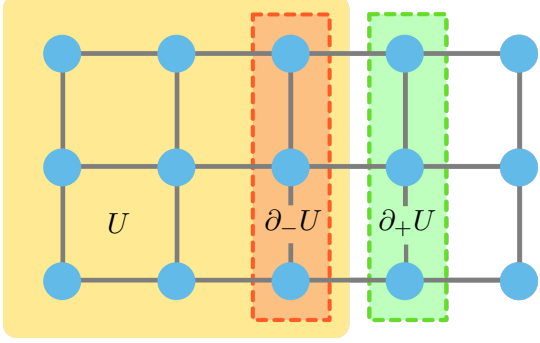
\begin{figure}[h]
	\centering
	\begin{tikzpicture}
		\node at (0,0) {\includegraphics[scale=0.5]{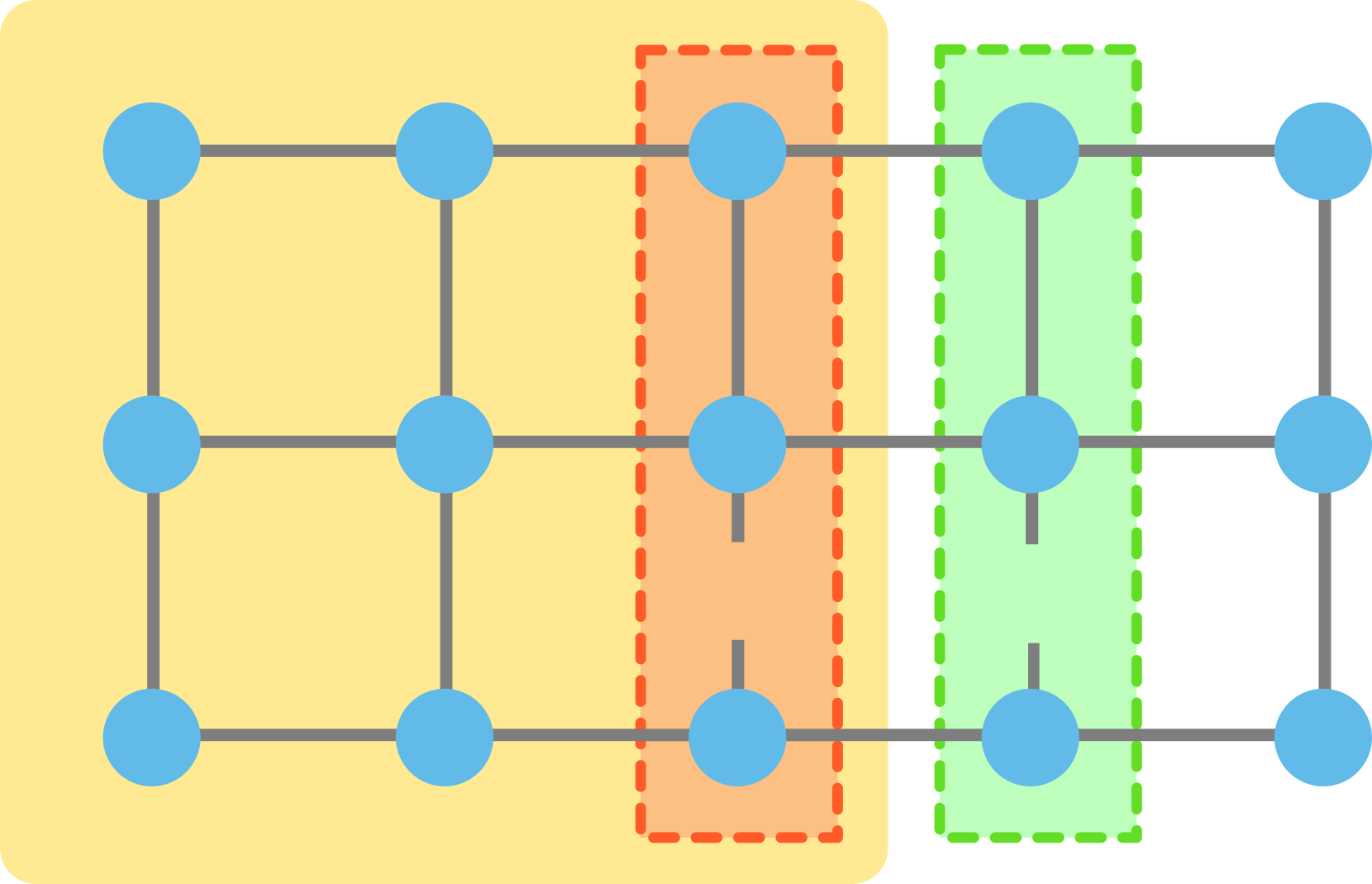}};
		\node at (-2,-0.75) {$U$};
		\node at (0.3,-0.75) {$\bdry_{-} U$};
		\node at (1.8,-0.75) {$\bdry_{+} U$};
		
	\end{tikzpicture}
	
	\caption{An illustration of inner and outer boundary of a subset of vertices. The subset $U$ is highlighted in yellow. The sets $\bdry_{-}U$ and $\bdry_{+} U$ are respectively highlighted in red and green. }
	\label{fig:inandoutbdry}
\end{figure}

\begin{definition}[Boundary]
	Let $\cC(\Omega)$ be a stabilizer code and $U \subset V$.
	We define the outer and inner boundaries respectively as follows:
	\begin{enumerate}
		\item \textbf{Outer boundary:} $\bdry_{+} U$ is the set of all qubits corresponding to $v \in \comp{U}$ such that there exists at least one stabilizer generator $\ssS$ and $u \in U$ satisfying $v,u \in \supp(\ssS)$.
		\item \textbf{Inner boundary:} $\bdry_{-} U$ is $\bdry_{+} \comp{U}$, the outer boundary of the complement of $U$.
	\end{enumerate}
\end{definition}

It follows from the Cleaning Lemma that if $U$ is correctable, then any logical $\ssL \in \cL$ can be cleaned to $\comp{U}$.
The converse holds true.

\begin{definition}[Decoupled subsets]
	\label{def:decoupled}
	Disjoint subsets $\{U_i\}, U_i \subset V$ are said to be \emph{decoupled} if no generator overlaps with more than one $U_i$, i.e.\ $\forall i, \ \forall j \neq i$,
	\[
	\bdry_+ U_i \cap U_j = \emptyset ~.
	\]
\end{definition}

\begin{lemma}[Union Lemma]
	\label{lem:bptunion}
	Let  $\{U_i\}_{i}$ be decoupled sets of qubits and write $T = \cup_i U_i$.
	If $\ssL \in \cL$ such that $\supp(\ssL) \subseteq T$, then $\ssL$ can be decomposed as a product of logicals, each supported entirely on one subset $U_i$: $\ssL = \otimes_i \ssL_{U_i} $, $\supp(\ssL_{U_i}) \subset U_i$, $\ssL_{U_i} \in \cL$.
	If each $U_i$ is correctable, it follows that $T$ is correctable.
\end{lemma}
\begin{proof}
	Let $U_i$ be one cluster of the decoupled set and $\ssL_{U_i}$ the restriction of $\ssL$ on $U_i$.
	$\ssL_{U_i}$ has to commute with any generators whose overlap with $\ssL$ is only contained within $U_i$.
	Due to the decoupling condition, this is the case for all generators having support on $U_i$.
	Since $\ssL_{U_i}$ commutes with all generators having support on $U_i$, we conclude $\ssL_{U_i} \in \cL$.
	However, every $U_i$ is correctable, therefore $\ssL_{U_i} = \ssI_{U_i}$, i.e.\ they act as identity on the qubits $U_i$.
\end{proof}

\begin{lemma}[Expansion Lemma]
	\label{lem:bptexpansion}
	Given correctable regions $U$, $T$ such that $ U \supset \bdry_+ T$, then $T \cup U$ is correctable.
\end{lemma}
\begin{proof}
	Set $W \equiv \comp{U \cup T}$.
	Since $U$ is correctable, any logical operator $\ssL$ whose support intersects with $U$ can be cleaned to $T \cup W$, and we let $\ssL'$ denote the cleaned operator.
	This implies that $\ssL' = \ssL_T \otimes \ssL_W$.
	Note that no check acts on both $T$ and $W$ as $\bdry_+ T \subseteq U$ which means that $\ssL_T$, by itself, is a bonafide logical operator.
	Since $T$ is correctable, we have $\ssL_T = \ssI$.
	Therefore $T \cup U$ is correctable.
\end{proof}

In addition to the algebraic view presented above, quantum codes can also be represented graphically.
We shall use the following object, called a connectivity graph representation. This representation depends on the \emph{generators} of a code, not on the code itself.

\begin{definition}[Connectivity graph]
	\label{def:connectivity}
	Let $\cC(\Omega)$ be a stabilizer code.
	Then the connectivity graph $G(\Omega) = (V, E)$ associated with $\Omega$ is defined so that:
	\begin{enumerate}
		\item $V = [n]$, i.e. each vertex is associated with a qubit, and
		\item $(u,v) \in E$ if and only if there exists a generator $\ssS \in \Omega$ such that $u,v \in \supp(\ssS)$.
	\end{enumerate}
	Further, $\bdry_{+}$ and $\bdry_{-}$ extend naturally to $G(\Omega)$
\end{definition}

For the sake of readability, we will refrain from referring to $\Omega$ explicitly. When we refer to a connectivity graph $G$ of a code $\cC$, it is to be understood that there exist a generating set $\Omega$ such that $\cC = \cC(\Omega)$ and $G = G(\Omega)$. 

\textbf{Remarks:}
We first remark that this representation is not a function of the code, as different generating sets can yield the same code but different graphs. 
Further the mapping from $\Omega$ to $G$ is not injective: different codes might yield the same connectivity graph, if the right generating sets are chosen.
This stands in contrast with the Tanner graph, another common graphical representation of LDPC codes.
However, despite this lack of uniqueness, the connectivity graph suffices to obtain the desired bounds on the distance and code dimension.
Note that this representation dispenses with Pauli labels between the stabilizers and qubits and no longer carries information concerning the commutation relations.
Taking this information into account could be important in restricting the types of graphs that emerge; we do not do so here.
This representation was also used for different purposes, see for example \cite{kovalev2013fault} or \cite{gottesman2014fault}.

The following observations will be useful.
Consider two disjoint subsets $U_1, U_2 \subset V$.
If there is no edge between $U_1$ and $U_2$ then they are decoupled.
Equivalently, the distance on the connectivity graph between these two sets is at least two.
In other words, $\bdry_{+} U_1$ is the neighborhood of $U_1$ in the connectivity graph.

If the quantum code family is LDPC, then the connectivity graph has bounded degree.
Suppose $\scrC = \{\cC_n\}$ is a code family with qubit degree upper bounded by $\delta_V$ and stabilizer degree bounded by $\delta_C$.
Then each vertex in the connectivity graph is connected to at most $\delta_V (\delta_C - 1)$ other qubits.
We expect the degree to be less than this because the stabilizer generators can overlap, and likely will, to obey commutation relations.

\textbf{Example:} As an example, consider a portion of the surface code as shown in fig.~\ref{fig:connectivity-toric} below.
The surface code is a code defined on the $2$-dimensional square grid.
The qubits are identified with the edges of the lattice, the $X$ stabilizers with the vertices and the $Z$ stabilizers with the faces.
An $X$ ($Z$) stabilizer acts on a qubit if the vertex (face) corresponding to the stabilizer is adjacent to the edge corresponding to the qubit.   
The corresponding connectivity graph has vertices on all edges of the grid in addition to diagonal connections.

\begin{figure*}
	\centering
	\begin{subfigure}[t]{0.3\textwidth}
		\centering
		\includegraphics[width=0.8\columnwidth]{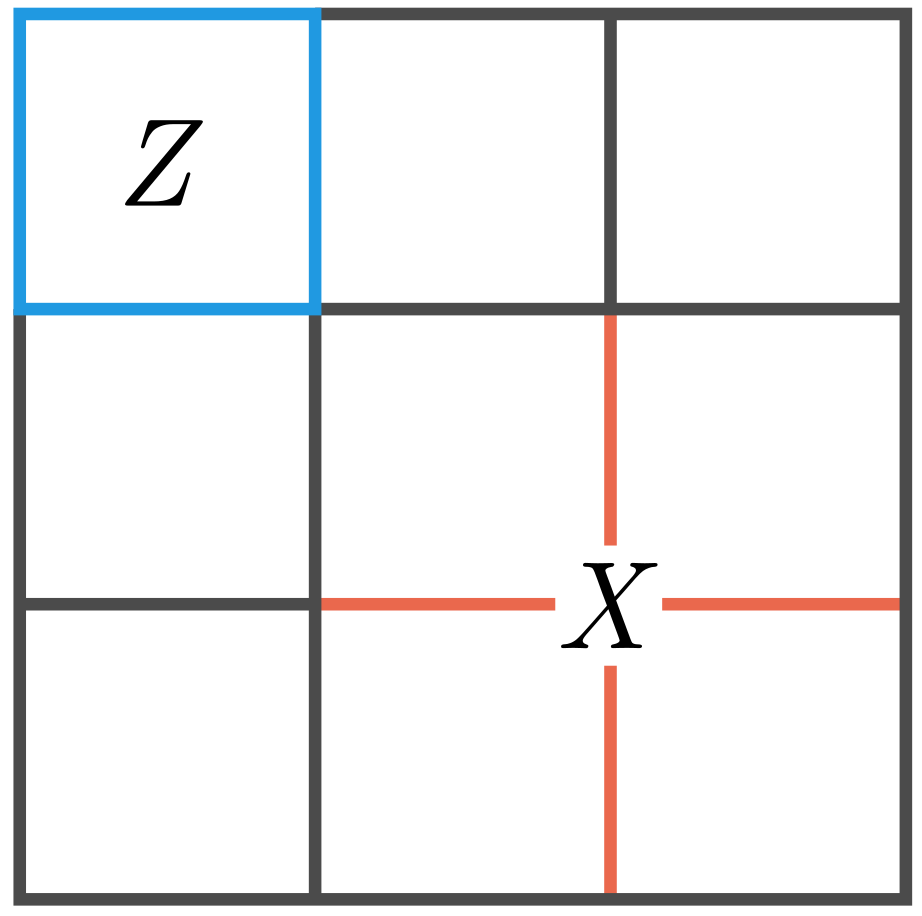}
	\end{subfigure}%
	~ 
	\begin{subfigure}[t]{0.3\textwidth}
		\centering
		\includegraphics[width=0.8\columnwidth]{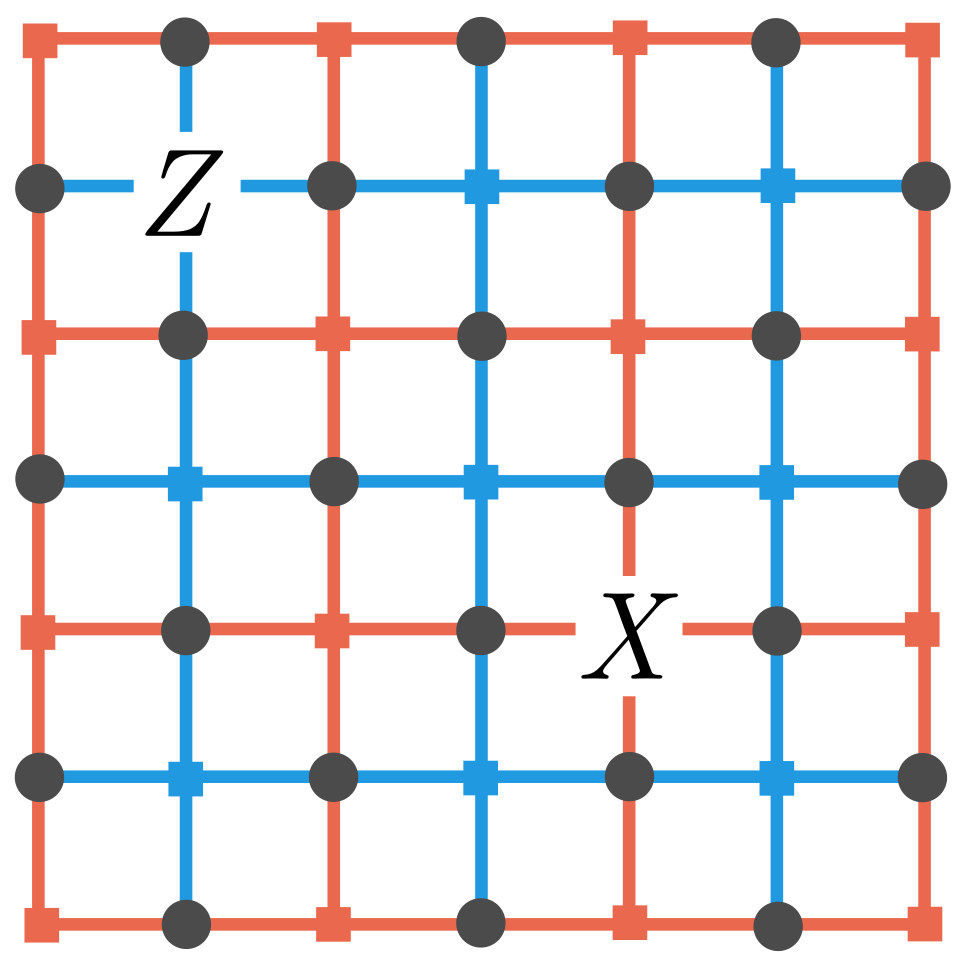}
	\end{subfigure}
	~ 
	\begin{subfigure}[t]{0.3\textwidth}
		\centering
		\includegraphics[width=0.8\columnwidth]{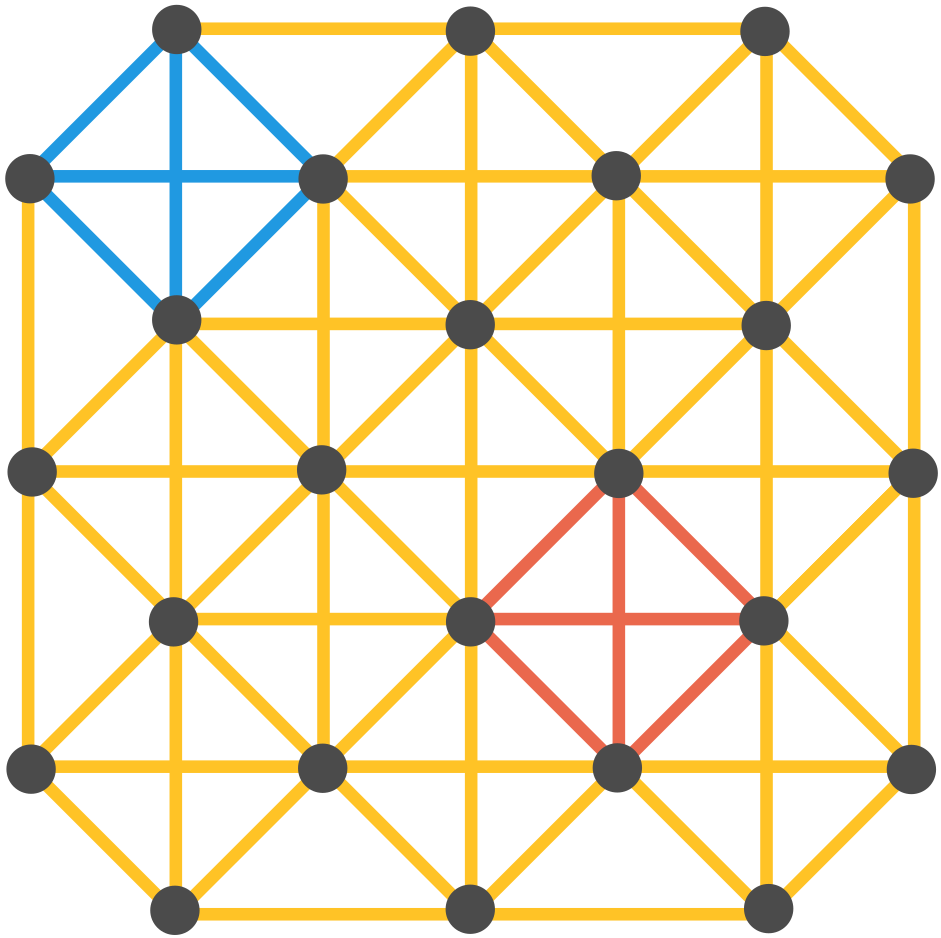}
	\end{subfigure}
	\caption{(a) Cellular representation of the surface code.
		$Z$-generators are associated with the faces, qubits are on the edges, and $X$-generator are on the vertices.
		(b) Tanner graph of the surface code.
		$Z$($X$)-checks are the blue (red) squares, and qubits are the grey dots.
		If two qubits are in the support of a generator, they are connected by an edge.
		(c) Connectivity graph of the surface code.
		The grey dots are still the qubits, but the generators are no longer represented.
		Instead, qubits share an edge when they are touched by the same separator.
		The blue and red edges are induced by the Z and X generators highlighted in (a) and (b).}
	\label{fig:connectivity-toric}
\end{figure*}

\subsection{Separator and treewidth}
\label{subsec:sep-treewidth}
In this section, we start by formally introducing the notion of a separator and other related metrics.
We then introduce a closely related metric, the treewidth.
This additional metric is introduced for a technical reason, as it will be the basis for the proof of Theorem \ref{thm:distancebnd}.
Fortunately, the separation profile and the treewidth are closely related, as shown by Lemma \ref{lem:twandsep}: this will allow us to restate our bound on the distance from Lemma \ref{lem:twdist} in terms of the separators. 

\begin{definition}
	Let $G = (V, E)$ be a graph, $\alpha \in [2/3,1)$.
	Then the $\alpha$-separator of $G$, written $\sansserif{sep}^\alpha(G)$ is the smallest set $S \subset V$ such that
	\begin{enumerate}
		\item $A$, $S$, $B$ are a disjoint partition, i.e.\ $V = A  \sqcup  S  \sqcup  B$.
		\item Both of $|A|, |B| \leq \alpha |V|$.
		\item There are no edges between $A$ and $B$.
	\end{enumerate}
	
\end{definition}

The separator might not be uniquely defined, as multiple sets could have the same size and still split the graph in two disjoint subgraphs.
However, this multiplicity does not affect our results, as it will suffice to prove the existence of \emph{one} such small set.
It will be useful to note that for any $\alpha$, we have $|\sansserif{sep}^\alpha(G)| \leq (1-\alpha)|V|$, as any set of such size naturally induces two sets, $A = \emptyset$, and $B$ such that $|B| = |V| - (1-\alpha)|V| = \alpha |V|$.

Consider the graph made of the disjoint union $G = G_1 \cup G_2$ where $G_1$ and $G_2$ are densely connected.
The separator of $G$ alone would only provide superficial information about its connectivity.
To make a more consistent statement about the connectivity of a graph, we introduce a notion of separability that also relies on subgraphs.

\begin{definition}
	\label{def:alphasep}
	For any graph $G$ on $n$ vertices, we define its $\alpha$-separation profile $s_G^\alpha : [1,...,n] \rightarrow \mathbb{N}$,
	\[
	s_G^\alpha(r) = \max\{ |\sansserif{sep}^\alpha (H)| : H \subseteq G, |H| \leq r\}~.
	\]
	
	To a set of graphs $\cG = \{G_n\}_n$, we associate a set of $\alpha$-separation profiles $\{s_n^\alpha\}_n$ where $s_n^\alpha$ is the separation profile of $G_n$.
\end{definition}

\textbf{Remark:} Since $s^\alpha_G(n) =  \Theta(s^{2/3}_G(n))$ --- see Lemma \ref{lem:salpha} --- our results do not rely on a specific value of $\alpha$, hence we will often write the separation profile $s_G$ or $s_n$ to mean $s_G^\alpha$ or $s_n^\alpha$.

As an example of separation profile, consider a grid graph as shown in fig.~\ref{fig:connectivity-toric}.
This graph is poorly connected; we can partition the vertices into two sets by removing a thin strip from the middle.
In other words, an $L \times L$ grid with $n = \Theta(L^2)$ vertices has a separator of size $L = \Theta(\sqrt{n})$.
As shown in the famous theorem by Lipton and Tarjan \cite{lipton1979separator}, this is true for any planar graph, more precisely for any $G$ planar, we have $s_G(r) =  O(\sqrt{r})$.
Such a bound can also be found for other classes of graphs, for example, the size of the separator for `local' graphs embedded in $D$-dimensional Euclidean space is known to be $O(n^{(1-1/D)})$ \cite{miller1997separators}.
In contrast, expander graphs famously have large separators, i.e.\ $s_G(r) = \Theta(r)$. 
For general graphs, there exist polynomial-time algorithms to \emph{approximate} their separator up to constant factors \cite{orecchia2012lorenzo}.

For many classes of graphs \cite{kawarabayashi2010separator,kisfaludi2020hyperbolic, dujmovic2015genus, gladkova2020separation, coz2020separation, lipton1979separator, miller1997separators} the separation profile of a graph is upper bounded by a polynomial up to constant factors, which is particularly amenable to recursive separation -- see Lemma \ref{lem:solrecursion} -- which will be essential to the formulation of our results. However the separation profile does not always assume a polynomial form, so to capture these cases, it will be helpful to consider the quantity $c(r)$ such that $s_G(r) = r^{c(r)}$. 
This motivates the following definition.

\begin{definition}
	\label{def:cminmax}
	Let $\scrC = \{\cC_n\}_n$ be a family of $\dsl n,k(n),d(n)\dsr$ quantum LDPC codes with non-trivial connectivity graphs $\cG = \{G_n\}_n$ with associated separation functions $\{s_n\}_n$.
	Consider the quantity $c_n(r) \equiv \log_r(s_n(r))$.
	For each $G_n \in \cG$, define $c_{\max}(n) =  \max_{r \in [d(n),n]} c_n(r)$.
\end{definition}

The quantity $c_{\max}$ measures how tightly the graph is connected by considering subgraphs of whose size lies in the interval $[d,n]$.
Consider fig.~\ref{fig:cmax} which shows a connectivity graph $G$ and subgraphs $\{H_i\} \subseteq G$.
Each of these subgraphs $H_i$ is an expander graph, i.e. $s_{H_i}(r) = \Theta(r)$.
\begin{figure}[h]
	\centering
	\includegraphics[width=\columnwidth]{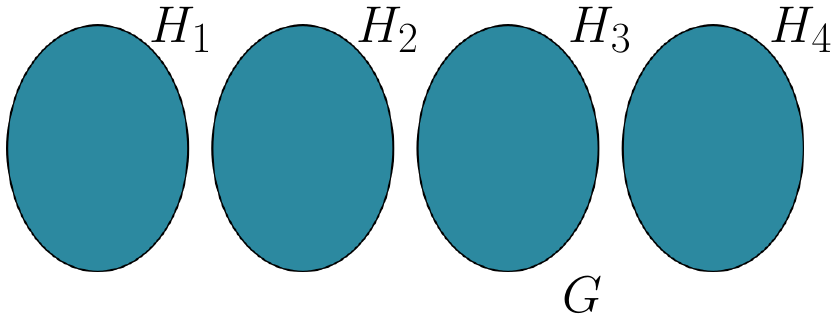}
	\caption{Visualizing $c_{\max}$ for a connectivity graph $G$.
		The connectivity graph $G$ is made up of several disconnected expander subgraphs of size $n^{\alpha}$; there are $n^{1-\alpha}$ such subgraphs, and therefore $G$ has $n$ vertices in total.
		This implies that $c_{\max} \approx 1$.}
	\label{fig:cmax}
\end{figure}
$G$ itself is comprised of $n^{1-\alpha}$ such disconnected expander graphs as shown.
Each subgraph $H_i$ has size $n^{\alpha}$, and so $G$ has $n$ vertices in total.
Suppose it was known that $d = O(n^{\alpha})$.
If we consider small enough subgraphs i.e.\ of size lesser than $n^{\alpha}$, then there exist subgraphs with large separators.
However, if we let $r = n^{\alpha}$, the largest separator corresponds to any subgraph $H_i$ and therefore $c_{\max} = 1$.

For the sake of readability, we will simply write $k \equiv k(n)$ and $d \equiv d(n)$.
By the definition above, there exists a subgraph $H \subseteq G_n$ such that $d \leq |H| \leq n$ and $|\sep(H)| = |H|^{c_{\max}}$.
We also note that for all $d \leq r \leq n$, we have that $s_n(r) \leq r^{c_{\max}(n)}$.
Further one can note that since $|\sep^\alpha(G)| \leq (1-\alpha)|V|$ for any graph $G$ then $s_n(r) \leq (1-\alpha)r $. Therefore we always have $c_{\max}(n) \leq\log_r(1-\alpha) + 1 \leq \log_n(1-\alpha) + 1 < 1$.

We now define the tree decomposition and then the treewidth. 
The tree decomposition of a graph $G$ is a tree whose nodes are clusters of vertices of $G$.
The width of the tree is the minimum size of its nodes and is yet another way to measure the connectivity of a graph.

\begin{definition}
	\label{def:treedecomp}
	A tree decomposition of a graph $G=(V,E)$ is a pair $(\{ Q(i), i \in \cI \}, \cT =(\cI,\cE))$ where $\{i : i \in I\}$ is a family of subsets $Q(i) \subseteq V$ and $\cT =(\cI,\cE)$ is a tree.
	The sets $\cI$ and $\cE$ refer to the nodes and edges of the tree $\cT$ respectively.
	Furthermore, the pair $Q, \cT$ must obey the following properties:
	\begin{enumerate}
		\item \label{it:vertex-containment} $\bigcup_{i \in \cI} Q(i) =V$,
		\item \label{it:edge-containment} for every edge $\{v,w\} \in E$ there exists $i\in \cI$ with $\{v,w\} \subseteq Q(i)$,
		\item \label{it:connectivity} for every $i,j,k \in \cI$ the following holds: if $j$ lies on the path from $i$ to $k$ in $\cT$, then $Q(i) \cap Q(k) \subseteq Q(j)$.
	\end{enumerate} 
	The \emph{width} of the tree decomposition $(\{Q(i) : i \in \cI \}, \cT = (\cI,\cE))$ is defined as $\max_{i \in \cI} |Q(i)| - 1$.
	The treewidth $\tw(G)$ of $G$ is the minimum width of a tree decomposition of $G$.
\end{definition}

To avoid confusion between the graphs $G$ and $\cT$, we shall henceforth refer to the \emph{vertices} $v \in V$ of $G$ and the \emph{nodes} $i \in \cI$ of $\cT$.
The tree decomposition of a graph is not unique; for instance, a trivial decomposition of a graph $G$ is to make one giant node $N$ containing all the vertices of $G$.
The treewidth, however, is the minimum width across all decompositions and is therefore well defined.
The notation $Q(i)$ is meant to be suggestive as it will soon refer to the qubits in that node.

\textbf{Example:} We consider some examples to illustrate this idea.
The first example is the tree decomposition of a tree as shown in fig.~\ref{fig:tree-treedecomp}.
In it is a binary tree of depth $2$ and the corresponding tree decomposition.
The vertices of the graph are gray and the nodes of the tree are green.
It is simple to check Property \ref{it:vertex-containment}.
Notice that every node of the tree decomposition contains a single edge from the tree trivially satisfying Property \ref{it:edge-containment}.
Finally, Property \ref{it:connectivity} is simple to verify: only two adjacent nodes ever share qubits.
Since the size of each node is $2$, the treewidth is $1$.
\begin{figure}[h]
	\centering
	\includegraphics[width=\columnwidth]{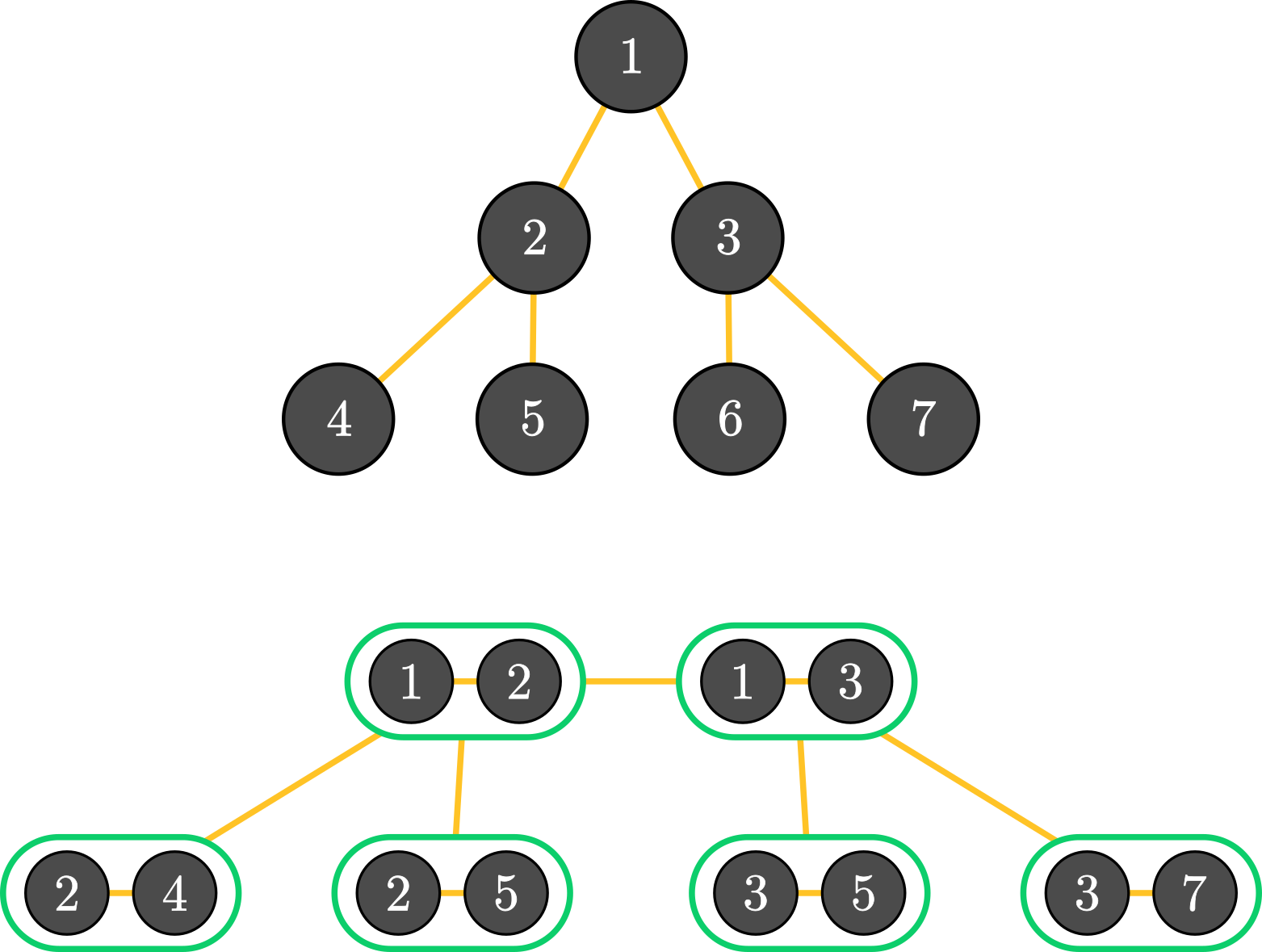}
	\caption{A tree above and its tree decomposition below.
		Each vertex of the tree is denoted using a gray circle and each node of the tree is denoted using a green box.
		The treewidth of a tree is $1$.}
	\label{fig:tree-treedecomp}
\end{figure}

As a second example, consider the surface code again as shown in fig.~\ref{fig:surface-treedecomp} on the left along with its tree decomposition on the right.
The graph on the left indicates via green boxes how vertices are partitioned to form the nodes of the tree.
Recall the structure of the connectivity graph of the surface code as shown in fig.~\ref{fig:connectivity-toric}.
We choose the nodes of the tree by selecting vertices of the connectivity graph diagonally as shown.
Again, it is straightforward to verify that this tree decomposition satisfies the definition.
First, the diagonals contain every vertex and thus satisfy Property \ref{it:vertex-containment}.
It is also straightforward to verify that every edge is contained in at least one node, satisfying Property \ref{it:edge-containment}.
Finally, since two successive nodes overlap on one diagonal array of vertices, the decomposition satisfies Property \ref{it:connectivity}.
The treewidth of this graph is obtained from the largest node of the tree which corresponds to the node $cd$.
The treewidth is therefore $11$.

In practice, there is a lot of interest in efficient algorithms to compute the treewidth of arbitrary graphs \cite{dell2018pace}.  It can be noted that since the treewidth and the size of the separator of a graph are within a constant of each other -- see Lemma \ref{lem:twandsep} -- the algorithm of \cite{orecchia2012lorenzo} can be used to obtain polynomial time approximation of the treewidth.

\begin{figure}[h]
	\centering
	\includegraphics[scale=0.5]{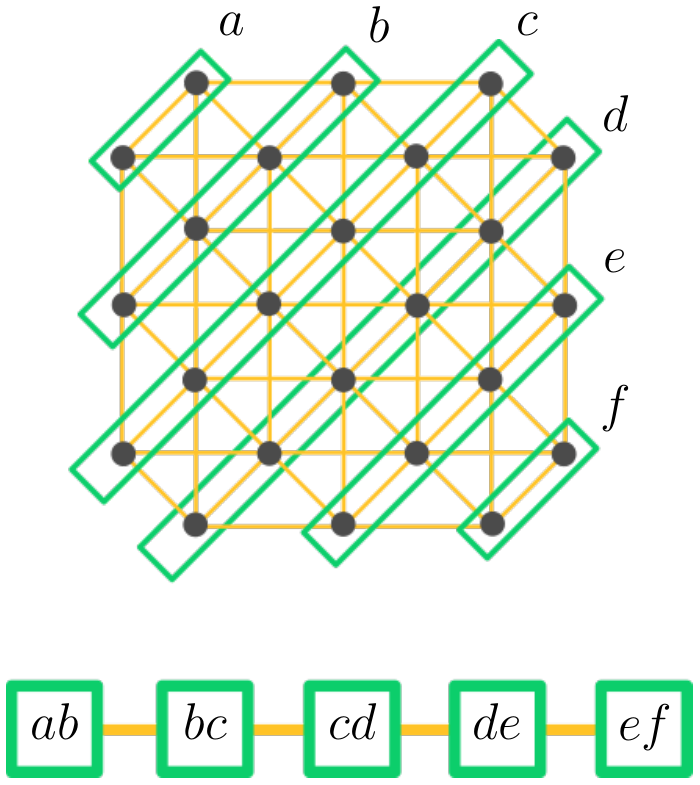}
	\caption{Above is the connectivity graph representation of the $3 \times 3$ surface code.
		Its vertices are partitioned into subsets $a,b,c,d,e,f$.
		Below is a tree decomposition of the same graph.
		The node $ab$ contains the union of the vertices in $a$ and $b$.}
	\label{fig:surface-treedecomp}
\end{figure}

As promised earlier, the separation profile and the treewidth are closely related.

\begin{lemma}
	\label{lem:twandsep}
	For any graph $G = (V,E)$ on $n$ vertices with separation profile $s_G$, then $s_G(n) =  \Theta(\tw(G))$.
	
\end{lemma}
\begin{proof}
	First, it was noted in \cite{bottcher2010bandwidth} that for a graph $G$ on $n$ vertices, $s^{2/3}_G(n) \leq \tw(G) + 1$, or $s^{2/3}_G(n) = O(\tw(G))$. In \cite{dvorak2019treewidth}, it is shown that $\tw(G) =  O(s^{2/3}_G(n))$.
	Since $s^{\alpha}_G(n) = \Theta(s^{2/3}_G(n))$, see Lemma \ref{lem:salpha}, then we have $s^{\alpha}_G(n) = \Theta(s^{2/3}_G(n)) = \Theta(\tw(G))$.
\end{proof}

\subsection{Linear treewidth, separation, and expansion}

Another commonly used notion of connectivity is that of expansion, which has already found applications in analyzing the structure of quantum codes \cite{fawzi2018efficient}.
Further expander graphs are an essential tool in classical error correction \cite{sipser1994expander,zemor2001expander}.
It is then natural to ask how the treewidth, the separability, and our results in general, relate to the expansion of a graph.

The vertex expansion of a graph is generally defined through its Cheeger constant.

\begin{definition}
	For any graph $G$ on $n$ vertices, we define its Cheeger constant $h(G)$ as
	\[
	h(G) = \min_{A \subset G, |A| \leq n/2} \frac{|\bdry A|}{|A|}~.
	\]
\end{definition}

\begin{lemma}
	\label{lem:expansion}
	For any graph $G$ on $n$ vertices we can find $H \subseteq G$ with $|H| \geq c' \cdot n$ and $h(H) \geq c/2$ for some $c,c' > 0$ if either of these two conditions is fulfilled 
	\begin{enumerate}
		\item $\tw(G) \geq c \cdot n$, or
		\item $s_G \geq c \cdot n + 1$.
	\end{enumerate}
\end{lemma}
\begin{proof}
	Case $1$ is Proposition 2 in \cite{grohe2009treewidth}.
	Case $2$, is readily obtained from Case 1 by the fact that $s^{\alpha}_G \leq s^{2/3}_G \leq \tw(G) + 1$ \cite{bottcher2010bandwidth} .
\end{proof}

\section{Main results}
\label{sec:main}

\subsection{Bound on the distance}
\label{subsec:distancebnd}
We now state and prove the first main result: the distance of a code family is bounded by the treewidth of the connectivity graph.

\begin{lemma}
	\label{lem:twdist}
	Let $\cC$ be a code and $G$ an associated connectivity graph of bounded degree $\delta$.
	If $G$ has treewidth $\tw(G)$ then the distance obeys $d \leq \delta \cdot (\tw(G) + 1)$.
\end{lemma}
\begin{proof}
	Consider a tree decomposition $\cT$ of $G$ such that the width of the tree $\cT$ is the treewidth of $G$.
	For the sake of contradiction, assume $d > \delta(\tw(G) + 1)$.
	
	Suppose the tree $\cT$ is non-trivial and has depth $\depth \geq 1$ (and the root at depth $0$).
	Let $p \in \cI$ be some node at depth $\depth - 1$.
	Let the leaves $\{j_{1},...,j_{t}\} \subset \cI$ be the children of $p$. 
	
	Consider the set $\cA = \union_i Q(j_i) \setminus \cup_i \bdry_+ Q(j_i)$.
	The purpose of $\cA$ is to be a correctable anchor whose boundary will be provably small.
	This will allow us to grow $\cA$ to a larger, but still correctable, region.
	
	First, it follows from lemma $\ref{lem:bptunion}$ that $\cA$ is itself correctable.
	This is because once the boundaries are removed, $\cA$ is a union of decoupled sets as per definition \ref{def:decoupled}. \footnote{Including the case of $Q(j_i)$ and $Q(j_k)$ sharing a qubits $q$. In that case, $ Q(j_i)\setminus \bdry_+ Q(j_k) \cup Q(j_k)\setminus \bdry_+ Q(j_i)$ can be decomposed as the union of three decoupled sets: something in $Q(j_i)$ not connected to anything in $Q(j_k)$, $q$, and something in $Q(j_k)$ not connected to anything in $Q(j_i)$. The Cleaning lemma still applies.}

	\begin{figure}[h]
		\centering
		\begin{tikzpicture}
			\node at (0,0) {\includegraphics[scale=0.2]{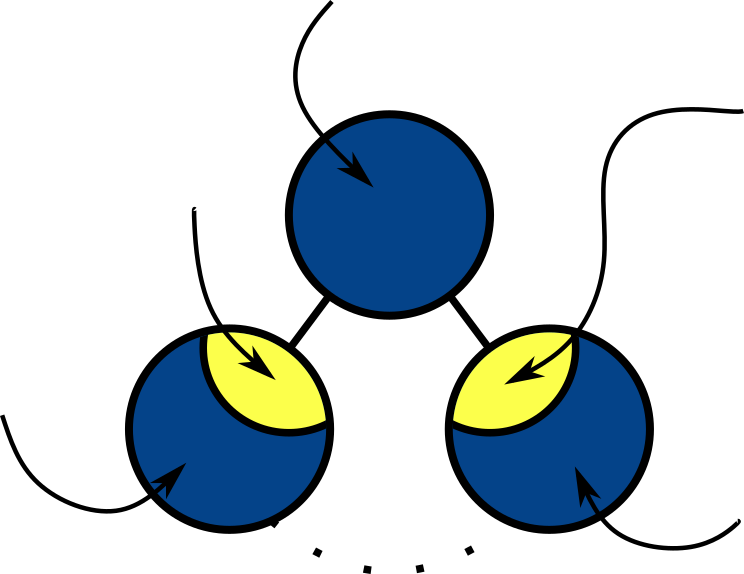}};
			\node at (-2,-0.4) {$Q(j_1)$};
			\node at (2.3,-1) {$Q(j_t)$};
			\node at (2.4,0.9) {$\bdry_{+} \cA$};
			\node at (-1.2,0.7) {$\bdry_{+} \cA$};
			\node at (0.1,1.7) {$Q(p)$};
		\end{tikzpicture}
		\begin{tikzpicture}
			\node at (7,0) {\includegraphics[scale=0.2]{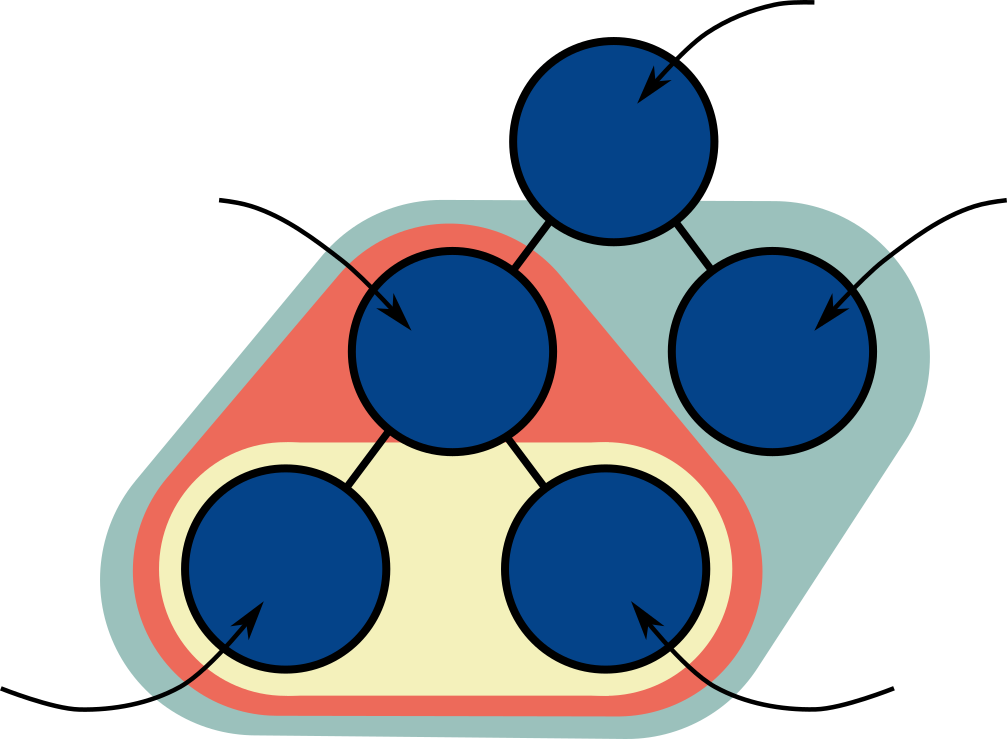}};
			\node at (4,-1.5) {$Q(j_1)$};
			\node at (9.3,-1.5) {$Q(j_2)$};
			\node at (4.9,0.9) {$Q(p_1)$};
			\node at (10.3,0.9) {$Q(p_2)$};
			\node at (9.3,2) {$Q(p_3)$};
		\end{tikzpicture}
		\caption{An illustration of the process used to iteratively grow a correctable region in the tree decomposition.
			The leaves $Q(j_1)\setminus \bdry_+ \cA$ and $Q(j_2)\setminus \bdry_+ \cA$ in the light yellow region can be verified to be correctable from the Union Lemma.
			Since their boundary is upperbounded by the size of the parent node $Q(p_1)$, we can find the nodes in the red triangle to be correctable using the Expansion Lemma.
			Call the union of these nodes $p_{red}$; $p_{red}$ and $p_2$ are children of the parent node $p_3$ and we may proceed recursively.}
		\label{fig:dstthmgrowing}
	\end{figure}
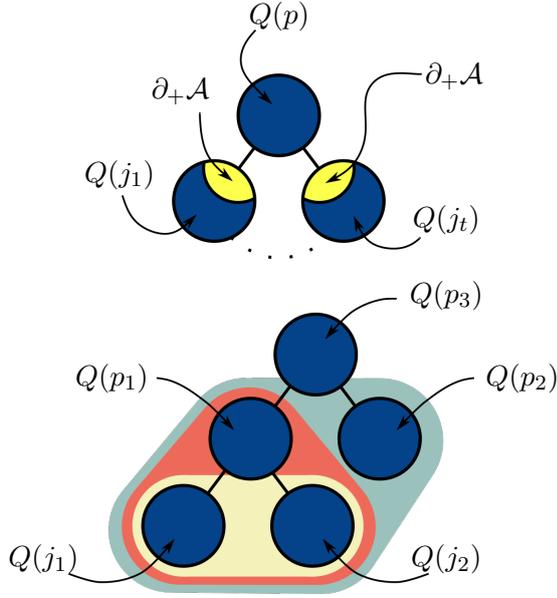
	
	Next, we turn to $\bdry_{+} \cA$.
	Consider any pair of qubits $u, u'$ such that $u \in \cA$, $u' \in \bdry_+ \cA$.
	By construction, there must be some leaf $j_i$ such that $u \in Q(j_i)$ and since $u' \in \bdry_{+} \cA$, either:
	\begin{enumerate}
		\item $u' \in \bdry_+ Q(j_i)$, or
		\item $u' \in Q(j_i)$ but $u' \in \bdry_+ Q(j_j)$ for some $j$.
	\end{enumerate}
	We conclude that $\bdry_+ \cA \subset \cup_i \bdry_+ Q(j_i)$.
	
	Let $\cP_{\ext} = Q(p) \cup \bdry_+ Q(p)$ be the extended parent.
	The purpose of $\cP_{\ext} $ will be to bound the size of $\bdry_+ Q(j_i)$ and extend our anchor. 
	By the definition of the treewidth, $|Q(p)| \leq \tw(n) + 1$.
	Furthermore, since the degree of the qubits in the connectivity graph is at most $\delta$, it follows that $|\bdry_+ Q(p)| \leq \delta (\tw(n)+1)$.
	Since $d > \delta (\tw(n) + 1)$ by assumption, both $Q(p)$ and $\bdry_{+} Q(p)$ are correctable.
	It follows from Lemma \ref{lem:bptexpansion} that $\cP_{\ext}$ is correctable.
	
	For every leaf $j_i$ and every $u \in \bdry_{-} Q(j_i)$, $u$ shares an edge with the exterior of $Q(j_i)$.
	Therefore there exists another node $j$ such that $u \in Q(j)$ by Property \ref{it:edge-containment} of Definition \ref{def:treedecomp}.
	By Property \ref{it:connectivity} of Definition \ref{def:treedecomp}, it follows that $u \in Q(p)$.
	We conclude that $\bdry_- Q(j_i) \subset Q(p)$.
	For any $v \in \bdry_{+} Q(j_i)$, either $v$ is in $Q(p)$, or it is outside of $Q(p)$, though still in its boundary: $\bdry_{+} Q(j_i) \subset Q(p) \cup \bdry_{+} Q(p) = \cP_{\ext}$.
	More generally $\bdry_+ \cA \subset \cup_i \bdry_{+} Q(j_i) \subset \cP_{\ext}$.
	
	We now have the necessary ingredients to extend the anchor.
	Since $\cA$ and $\cP_{\ext}$ are correctable with $\bdry_+ \cA \subset \cP_{\ext}$, then $\cup_i Q(j_i) \cup Q(p) \subset \cA \cup \cP_{\ext}$ is correctable \footnote{To verify that the inclusion holds, it is sufficient to verify that $\cup_i Q(j_i)\subset \cA \cup \cP_{\ext} $. Since  $\cA = \union_i Q(j_i) \setminus \cup_i \bdry_+ Q(j_i)$ it suffices to show that $\cup_i \bdry_+ Q(j_i) \subset \cP_{\ext} $, which is the conclusion of the previous paragraph}.
	This too follows from Lemma \ref{lem:bptexpansion}.
	This shows that the qubits in $p$ together with its children together are correctable.
	We can combine these nodes to form one larger leaf.
	Notice that after combining the $p$ and its children into one node, the resulting tree is still a valid tree decomposition of the connectivity graph $G$.
	Save for the new amalgamated node, the size of the rest of the nodes of the tree is still upper bounded by $\tw(G) + 1$.
	
	The proof now proceeds by repeating this process until the entire tree is contracted to one node.
	First, we can contract the children of all nodes at depth $\cD-1$ to reduce the depth of the entire tree to $\cD-1$.
	It follows that this tree is also a valid tree decomposition, with all the leaves corresponding to correctable sets as proved above.
	This process can be iterated until the entire tree becomes one giant node which itself must be correctable.
	If the tree decomposition has several disjoint components, each of these components is a tree with bounded treewidth.
	Each can be proved to be correctable, and then since disjoint, their union is also correctable. 
	This implies the whole code is correctable, a contradiction if the code is to encode at least one logical qubit.
\end{proof}

From Lemma \ref{lem:twandsep}, it can be noted that easily separable graphs have bounded treewidth.
Applied to Lemma \ref{lem:twdist}, the proof of Theorem \ref{thm:distancebnd} follows naturally.

\begin{theorem}
	\label{thm:distancebnd}
	Let $\scrC = \{\cC_n\}$ be a family of $\dsl n,k,d \dsr$ quantum LDPC codes with associated connectivity graphs $\cG = \{G_n\}$ and associated separation profiles $\{s_n\}_n$. Then,
	\begin{align*}
		d = O(s_n(n)) ~.
	\end{align*}
	In particular if $s_n(r) = O(r^c)$ where $0 \leq c \leq 1$, then
	\begin{align*}
		d = O(n^c)~.
	\end{align*}
\end{theorem}
\begin{proof}
	Note that $d = O(\tw(G_n)) = O(s_n(n))$. The first bound is from Lemma \ref{lem:twdist}, the second is from Lemma \ref{lem:twandsep}.
\end{proof}

There are many classes of graphs for which the value of $c$ is known \cite{kawarabayashi2010separator,kisfaludi2020hyperbolic, dujmovic2015genus, gladkova2020separation, coz2020separation, lipton1979separator}, and it can be estimated for arbitrary families in polynomial time \cite{orecchia2012lorenzo}.

A useful property of separation profiles and treewidth is that both metrics are somewhat robust to the addition of edges in a graph.
This then leads to the following corollary.

\begin{corollary}
	\label{cor:addedges}
	Let $\cC$ be a quantum LDPC code on $n$ qubits with associated connectivity graph $G$ and separation profile $s_G$. 
	Then consider any code $\cC'$ such that its connectivity graph $G'$ corresponds to $G$ augmented with a set of edges: $G' = (V, E \cup E_{\text{aug}})$, then we have
	
	\[
	d = O(s_G(n) + |E_{\text{aug}}|)
	\]
\end{corollary}
\begin{proof}
	Any separator $S$ of the graph $G$ can be augmented to be a separator of the graph $G'$ by removing the vertices involved in the edges $E_{\text{aug}}$. 
	Since there are at most $2 |E_{\text{aug}}|$ of these vertices, then $s_{G'}(r) \leq s_{G}(r) + 2 |E_{\text{aug}}|$ and the result follows.
\end{proof}

A straightfoward consequence is that if a code has a poor connectivity, it takes a significant number of edges to overcome the associated poor distance. For example a planar LDPC has distance $d = O(n^{1/2})$, and to improve it to any $d' = \Omega(n^{1/2 + \epsilon})$ with $\epsilon > 0$ one needs to add at least $\Omega(n^{1/2+\epsilon})$ edges.

This theorem then leads to the following conclusion--unless the graph is very connected, the distance cannot grow linearly.
\begin{corollary}
	\label{cor:sepdist}
	If $\{\cC_n\}$ is a family of $\dsl n,k,d \dsr$ codes such that the associated separation profiles $\{s_n\}_n$ satisfy $s_n(n) = o(n)$. Then $d = o(n)$.
\end{corollary}

Conversely, any family $\scrC = \{\cC_n\}$ with linear distance \cite{panteleev2020quantum} implies the existence of an expander family of graphs $\{G'_n\}$, $G'_n \subset G_n$ by Lemma \ref{lem:expansion}.

\subsection{Bound on the code dimension}
\label{subsec:dimensionbnd}

In this section, we mirror the strategy of \cite{bravyi2010tradeoffs} to bound the dimension of a code with poor connectivity.
Before diving into the formal proofs, we outline the proof.

For a given quantum code $\cC$ on a set $V$ of qubits, it can be shown that if a subset of qubits $A \subset V$ is correctable, then $\cC$ satisfies $k \leq |V \setminus A|$.
Then, one can always pick $A$ such that $|A| = d -1 < d$, and obtain $k \leq n - d + 1$.
This is just a weaker version of the Singleton bound \cite{gottesman1997stabilizer}.
This bound does not rely on the connectivity of the code and tells us very little about the asymptotic behavior of quantum codes. 

In order to improve this bound, we will make use of the connectivity graph $G = (V,E)$. Consider a set $S \subset V$, such that $S$ is a separator of $G$ and therefore induces a tripartition $A_1 \sqcup V \sqcup A_2$. Another perspective is that removing $S$ from the graph $G$ induces two disjoint graphs $G_{A_1} = (A_1, E_{A_1})$ and $G_{A_2}= (A_2, E_{A_2})$.
Then, by Lemma \ref{lem:bptunion}, if $A_1$ and $A_2$ are correctable, since there are no edges between $A_1$ and $A_2$, then $A_1 \cup A_2$ is correctable.
This implies that $k \leq |V \setminus A_1 \cup A_2| = |S|$, which gives us a connectivity-dependent bound: if a graph has small separators, then the bound on $k$ can be expected to be restrictive.

This strategy can be extended to the case where $A_1$ and $A_2$ are not correctable. It suffices to ``take'' new separators in $G_{A_1}$ and $G_{A_2}$, until every induced subgraph is correctable -- which can be guaranteed when every subgraph is of size $d-1$. We then have $k \leq \sum_{\text{sum over the separators}} |S|$

We now proceed to make this statement formal.
We begin by restating the following result as a lemma and provide an alternate proof without the use of von Neumann entropies.
The tradeoff is that our proof only works in the case of stabilizer codes, whereas the original statement applies to all quantum codes.

\begin{lemma}[Bravyi-Poulin-Terhal \cite{bravyi2010tradeoffs}, Eq.~14]
	\label{lem:bptabc}
	Consider an $\dsl n,k,d\dsr$ stabilizer code $\cC$ defined on a set of qubits $Q$, $|Q| = n$, such that $Q = A \sqcup B \sqcup C$.
	If $A,B$ are correctable, then
	\begin{equation}
		k  \leq |C|~.
	\end{equation}
\end{lemma}
\begin{proof}
	Let $H \in \bbF_2^{m \times 2n}$ be the symplectic representation of the stabilizer generators of $\cC$.
	We let $\rank{H} = n-k$ denote the rank of $H$ and for some set of indices $E \subseteq [n]$, we let $H_{E}$ denote the matrix obtained by selecting those columns indexed by $E$.
	
	Delfosse and Z\'emor \cite{delfosse2013upper} show (see Lemma 3.3) that an erasure $E \subseteq [n]$ is correctable if and only if:
	\begin{align}
		\label{eq:dzcorrectable}
		2|E| \leq \rank(H) + \rank(H_{E}) - \rank(H_{\comp{E}})~.
	\end{align}
	Consider the tripartition $Q = A \sqcup B \sqcup C$ where $A$ and $B$ are correctable.
	We infer from eq.~\ref{eq:dzcorrectable} that
	\begin{align}
		2|A| &\leq \rank(H) + \rank(H_A) - \rank(H_{BC}) \nonumber \\
		&\leq \rank(H) + \rank(H_A) - \rank(H_{B})~.
		\label{eq:acorrectable}
	\end{align}
	The last inequality follows because removing columns from a matrix can at best reduce its rank.
	Similarly, we obtain
	\begin{align}
		\label{eq:bcorrectable}
		2|B| & \leq \rank(H) + \rank(H_B) - \rank(H_{A})~.
	\end{align}
	Adding eq.~\ref{eq:acorrectable} and eq.~\ref{eq:bcorrectable}, we get $(|A| + |B|) \leq \rank(H)$.
	We can now substitute $|A| + |B| = n - |C|$ (because $A$, $B$, $C$ form a tripartition of the set of qubits) and $\rank(H) = n-k$ into this equation to obtain $k \leq |C|$.
\end{proof}

Then, the bounds on the code dimension and the level of the transversal gates depend on how costly it is to partition a graph. The following definition and lemma formalize this affirmation.

\begin{definition}
	\label{def:cS}
	For a graph $G$ with separation profile $s_G$, the function $\cS_d$ is defined by the recurrence relation $\cS_d(r) = s_G(r) + \cS_d(\alpha r) + \cS_d((1-\alpha)r)$, together with the condition that $\cS_d(t) = 0$ for all $t < d$.
\end{definition}

\begin{lemma}
	\label{lem:recpartition}
	Let $G$ be a graph with separation profile $s_G$ on $n$ vertices. For every $d \leq n$, there exists a partition $V = A \sqcup \comp{A}$, with $A$ a union of disjoint subsets of size strictly less than $d$, and $|\comp{A}| \leq \cS_d(n)$.
\end{lemma}
\begin{proof}
	Let $W \subset V$, then we define $\text{cost(W)}$ as the size of the smallest set $J_W \subset W$ such that $K \equiv W \setminus J_W$ is a union of disjoint subsets of size strictly less than $d$. $J_W$ might not be uniquely defined, but as we are only interest in its size, there is no loss of generality.
	
	Now consider a separator $S_W$ of the subgraph induced by $W$. This separator provides us with a partition $W_L \cup S_W \cup W_R$.
	
	It is then easy to verify that $\text{cost}(W) \leq s_G(|W|) + \text{cost}(W_L) + \text{cost}(W_R)$. Indeed, since $W_R$ and $W_L$ are disjoint, then $K_{W_R} \subset W_R$ and $K_{W_L} \subset W_L$ are disjoint. Therefore $K_{W_R} \cup K_{W_L}$ is a union of disjoint subsets, all of which are of size strictly less than $d$. This then gives $\text{cost}(W) \leq |W \setminus K_{W_R} \cup K_{W_L}| = |S_W| + |J_{W_L}| + |J_{W_R}| = |S_W| + \text{cost}(W_L) + \text{cost}(W_R)$. Further, by the definition of the separation profile, we have $\text{cost}(W) \leq s_G(|W|) + \text{cost}(W_L) + \text{cost}(W_R)$.
	
	This upper bound on $\text{cost}(W)$ is not very tractable and cannot be unraveled in a practical way. To solve this issue, we define $\cT(r) \equiv \max_{W \subseteq V, |W| \leq r} \text{cost}(W)$, and we have $\cT(r) \leq s_G(r) + \cT(|W_L|) + \cT(|W_R|)$. It is then possible to verify -- see Lemma \ref{lem:solrecursion} -- that one has $\cT(r) \leq S_d(r)$, where $S_d(r)$ is defined as in Definition \ref{def:cS}.
	
	We can then take $|\comp{A}| \leq \text{cost}(V) \leq \cT(n) \leq \cS_d(n)$.
\end{proof}

We can summarize the general idea of our results as follows.
Note that to find large $A,B$, given an easily separable graph, we can recursively separate it to obtain small correctable regions, which will be $A$.
Then $G \setminus A$ can be recursively separated anew, yielding $B$.
See Figure \ref{fig:recursive-sep}.

\begin{figure}[h]
	\centering
	\includegraphics[width=\columnwidth]{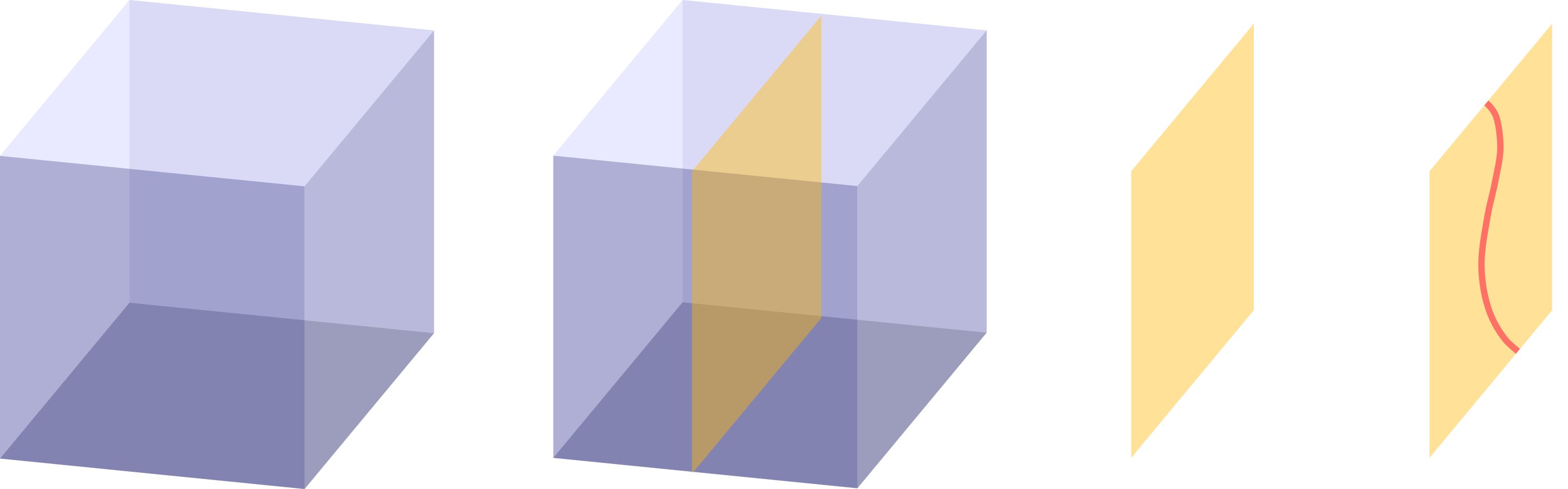}
	\caption{Graphical representation of our argument.
		The blue cube can be separated into two smaller correctable elements, which will be $A$.
		Then the separator, the yellow plane, can be separated again yielding $C$, the red line, and $B$ what remains of the yellow plane.}
	\label{fig:recursive-sep}
\end{figure}

\begin{lemma}
	\label{lem:sepbpt}
	Let $\cC$ be a code on $n$ qubits and $G = G(\cC)$ be an associated connectivity graph.
	If $\cS_d$ is defined as in \ref{def:cS}, then we have the bound $k \leq \cS_d \circ \cS_d(n)$.
\end{lemma}
\begin{proof}
	From Lemma \ref{lem:recpartition}, we can recursively separate the connectivity graph to find $\comp{A}$ such that $|\comp{A}| = \cS_d(|V|) = \cS_d(n) $.
	Since $\comp{A} \subset G$ then $s_G$ also bounds the size of the separators in the subgraph induced by $ \comp{A}$.
	Applying Lemma \ref{lem:recpartition} again to $ \comp{A}$, we can find sets $B$ and $\comp{B}$ such that $|\comp{B}| = \cS_d(|\comp{A}|) = \cS_d \circ \cS_d(n)$. 
	
	Next, note that $A$ is a union of disjoint subsets $\{V_\bullet\}$.
	By construction of the connectivity graph, this implies they share no generator, and are decoupled sets as per definition \ref{def:decoupled}.
	Since every subset $V_\bullet$ has size less than $d$, we can use the Union Lemma \ref{lem:bptunion} to show that $A$ is correctable.
	Similarly, $B$ is correctable.
	
	From Lemma \ref{lem:bptabc}, with $\comp{B}$ as the region $C$, we obtain $k \leq \cS_d \circ \cS_d(n)$.
\end{proof}

Ideally, one could then simply plug in $s^\alpha_G$ in Definition \ref{def:cS}, and obtain a closed form for the formula with Lemma \ref{lem:solrecursion} in Appendix \ref{app:closed-form-recurrence}.
However, the separation profile may not always be a polynomial.
In these instances, we can generalize this result with the following theorem.

\begin{theorem}
	\label{thm:dimbndcmax}
	Let $\scrC = \{\cC_n\}$ be a family of $\dsl n,k,d \dsr$ quantum codes with connectivity graphs $\cG = \{G_n\}$, and their associated separation functions $\{s_n\}_n$.
	Let $c_{\max}(n)$ and $r_0(n)$ be defined as in Definition \ref{def:cminmax}, then 
	\begin{align*}
		k = O(d^{2(c_{\max}(n)-1)}n \log(n)^2)~.
	\end{align*}
	
	Further, if we have $c_{\max}(n) \leq c_0$ for a constant $c_0 \in (0,1)$, then
	\begin{align*}
		k = O(d^{2(c_{\max}(n)-1)}n)~.
	\end{align*}
\end{theorem}

\begin{proof}
	By definition, we have $s_n(r) \leq r^{c_{\max}(n)}$. Applying Lemma \ref{lem:sepbpt} with Lemma \ref{lem:solrecursion} gives the desired result.
\end{proof}

This theorem tells us that a code with high $k,d$ has to have a very dense subgraph. 
Indeed, as we have $k \in \tilde{O}(n/d^{1-c_{\max}(n)})$, then for $k \sim n$, and $d \rightarrow \infty$, we need $c_{\max}(n)$ close to $1$.
Equivalently, for a code to achieve constant rate and growing distance, it needs to contain a dense subgraph.
Furthermore, this subgraph has to have density at least $d$.

In many cases, the family of connectivity graphs can be verified to satisfy $s_n(r) \in O(r^c)$, with $c < 1$, this is notably the case for many classes of local graphs \cite{teng1991phd}, as they exhibit a rather limited structure \cite{teng1998combinatorial}.
In that case the following corollary will be useful.

\begin{corollary}
	\label{cor:dimbndpoly}
	Let $\scrC = \{\cC_n\}$ be a family of $\dsl n,k,d \dsr$ quantum codes with connectivity graphs $\cG = \{G_n\}$, and associated separation profiles $\{s_n\}_n$.
	If $s_n(r) \in O(r^c)$ for some $c < 1$, then
	\[
	k =  O(d^{2(c-1)}n)~.
	\]
\end{corollary}
\begin{proof}
	The proof follows from applying Lemma \ref{lem:sepbpt} with Lemma \ref{lem:solrecursion}.
\end{proof}

As previously mentioned, there are many classes of graphs for which the value of $c$ is known \cite{kawarabayashi2010separator,kisfaludi2020hyperbolic, dujmovic2015genus, gladkova2020separation, coz2020separation, lipton1979separator}, and it can be estimated for arbitrary families in polynomial time \cite{orecchia2012lorenzo}..

Although this result allows us to formulate a bound in purely graph-theoretic terms, we cannot rederive the Bravyi-Poulin-Terhal bound.
Indeed, they make use of the Expansion Lemma to obtain regions $V_\bullet$ of size $d^2$, instead of $d$.
This optimization cannot be carried out here as we do not have a guarantee on the boundary of the regions $V_\bullet$ we create.
In other words, we make no assumptions on the structure of the subgraph that we obtain from the separator.
As a consequence, our bound applies to all codes and not just LDPC codes.
As we will see when dealing with $D$-dimensional hyperbolic codes, some spaces induce small separators \emph{and} large boundaries so we do not expect the Expansion Lemma to always yield a tighter bound.
These are highly nonlocal codes and may be able to bypass the restrictions on purely local codes.
These spaces might be expected to yield a poor distance \emph{and} better tradeoffs.
Without being able to pin down the structure of the subgraph induced by the separator, we cannot derive tighter bounds on the rate-distance tradeoff.

\subsection{Bounds on transversal gates}
\label{subsec:transversal}

In this section, we prove that transversal gates on quantum LDPC codes can only implement a limited set of transformations on the encoded information.
We begin with by recalling a result from \cite{bravyi2013classification} that we will build on.
\begin{lemma}
	\label{lem:bravyikoenig}
	Let $\cC$ be a code such that its set of qubits can be partitioned as $Q = \cup_{i = 1}^{i = R+1} \Lambda_i$, where each $\Lambda_i$ is correctable.
	The transversal gates on this code are limited to the $R$-th level $\cK^{(R)}$ of the Clifford hierarchy.
\end{lemma}

We are then in position to prove the following result.
\begin{lemma}
	\label{lem:septrans}
	Let $\cC$ be a code on $n$ qubits and $G = G(\cC)$ be an associated connectivity graph.
	Let $\cS_d$ be the function defined in Definition \ref{def:cS}, and let us denote $\cS_d$ composed $R$ times with itself as
	\[
	\circ_R \cS_d(n) = \underbrace{\cS_d \circ \cdots \circ \cS_d\,}_\text{$R$ times}(n)~.
	\]
	Then the tranversal gates on this code are limited to the $R$-th level of the Clifford hierarchy $\cK^{(R)}$ for the smallest $R$ satisfying
	\begin{align*}
		\circ_R \cS_d(n) < d~.
	\end{align*}
\end{lemma}
\begin{proof}
	We proceed iteratively.
	The recursive separation of Lemma \ref{lem:recpartition} acting on $G$ yields partitions $A_1$ and $\comp{A_1}$, where $A_1$ is correctable, and $|\comp{A_1}| \leq S_d(n)$.
	For $i > 1$, we separate $ \comp{A_{i-1}}$ into $A_i$ and $\comp{A_i}$ (note that $\comp{A_i}$ is the complement of $A_i$ within $A_{i-1}$ and not the entire graph $G$).
	We repeat this process $R$ times until $\comp{A_R} $ is correctable.
	In doing so we get $R+1$ correctable regions which is achieved when $\circ_R S_d (n) < d$.
\end{proof}

\begin{theorem}
	\label{thm:transpoly}
	Let $\scrC = \{\cC_n\}$ be an $\dsl n,k,d\dsr$ code family and $\cG = \{G_n\}$ be an associated connectivity graph, and $\{s_n\}_n$ the associated separation profiles.
	Suppose that $s_n(r) = O(r^c)$ for $c \in (0,1)$, and $d = \Theta(n^\alpha)$ for $\alpha \in (0,1)$.
	Then transversal gates on $\cC$ belong to $\cK^{(R)}$ where $R = \lceil \frac{1 - \alpha}{\alpha (1-c)}\rceil$.
\end{theorem}
\begin{proof}
	Since $s_n(r) = O(r^c)$, we known from Appendix \ref{app:closed-form-recurrence} that $S_d(n) \leq \sigma d^{c-1}n$ for some constant $\sigma > 0$, for sufficiently large $n$. 	
	The condition from Lemma \ref{lem:septrans} can be satisfied if 
	\[
	(\sigma d^{c-1})^R n < d~.
	\]
	Equivalently, taking the logarithm base $n$, we get
	\[
	R[\log_n(\sigma) + (c-1)\log_n(d)] < \log_n (d/n)~.
	\]
	Rearranging terms, we get
	\[
	R > \frac{1-\log_n(d)}{(1-c)\log_n(d) - \log_n(\sigma)}~.
	\]
	As we assume that $d \geq \sigma' n^\alpha$ for sufficiently large $n$, it is sufficient to satisfy
	\[
	R > \frac{1-\log_n(\sigma' n^\alpha)}{(1-c)\log_n(\sigma' n^\alpha)}~.
	\]
	Or, in the limit $n \to \infty$ 
	\[
	R > \frac{1 - \alpha}{\alpha (1-c)}~.
	\]
	Therefore $R = \lceil \frac{1 - \alpha}{\alpha (1-c)} \rceil$ is sufficient to obtain $R+1$ disjoint cleanable regions as stipulated in Lemma \ref{lem:septrans}.
\end{proof}

\textbf{Remarks:} In the case of $D$-dimensional quantum codes, it can be compared to the results of Pastawski and Yoshida \cite{pastawski2015fault}.
From \cite{miller1991unified,miller1997separators,teng1991phd}, a local graph in $\bbR^D$ satisfies $s_G(r) \in O(r^c)$ where $c = 1- 1/D$.
This then implies that
\begin{align*}
	R \leq \left\lceil \frac{1-\alpha}{\alpha}D \right\rceil~.
\end{align*}

While the bound from Pastawski and Yoshida, can be re-expressed to read
\begin{align*}
	R \leq \left\lceil (1-\alpha)D  + 1 \right\rceil~.
\end{align*}

There are instances where our bound may yield slightly better results; we thank Sam Cree for pointing this out.
For example, consider $D=3$ and $\alpha = 0.6 < 1-1/D = 2/3$.
Our bound implies that transversal gates must lie in the Clifford group, whereas the Pastawski-Yoshida bound implies that transversal gates are only contained in the third level of the Clifford hierarchy.
Of course, it is unclear whether such a code can be constructed.

We cannot reproduce the Pastawski-Yoshida bound for the same reason that we cannot reproduce the Bravyi-Poulin-Terhal bound:
in \cite{bravyi2013classification}, the separation of the $D$-dimensional lattice is better than the separation into multiple sectors we have based on $\cS_d$.

Indeed the interest of the results presented here lies within more exotic spaces and constructions where the lattice-based approach of the numerous no-go theorems in $\bbR^{D}$ breaks down.

We also mention a general limitation on obtaining practical codes.
Let $\scrC$ be a family of quantum LDPC codes with $s_n$ separation profiles.
Suppose $s_n(r) \propto r^c$ and that we can achieve $d = \Theta(n^c)$.
This implies that we can at best implement gates in $\cK^{(R)}$, where $R = \lceil 1/c \rceil$.
In particular, if $c > 1/2$, we are limited to Clifford gates.
This implies that there is a tradeoff between the distance and our ability to perform transversal gates even without the restriction of locality.

It is also interesting to note that it is shown in Burton and Brown \cite{burton2020limitations} that all hypergraph product codes are limited to the Clifford hierarchy, regardless of their separation profile. It raises the question of whether these codes offer the best trade-off between connectivity and versatility of transversal gates.
\section[]{Quantum codes in $\bbH^D$}
\label{subsec:hyperbolic}

Several constructions of quantum codes are naturally expressed through hyperbolic geometry \cite{delfosse2013tradeoffs,breuckmann2020singleshot,freedman2002z2}.
We use our results to study limitations of codes embeddable in $\bbH^D$ and on hyperbolic surfaces.
We demonstrate that $D$-dimensional hyperbolic codes have distance upper bounded by $O(n^{(D-2)/(D-1)})$, reminiscent of Euclidean codes in $(D-1)$-dimensions.
Interestingly, the tradeoff between the code dimension and distance is the same as that for local codes in $\bbR^D$.
Our results follow from some recent work by Kisfaludi-Bak \cite{kisfaludi2020hyperbolic} who proved that certain classes of hyperbolic graphs have bounded separators.

We begin by comparing our work with previous results to provide some intuition on what follows.
Recall that the Bravyi-Terhal and Bravyi-Poulin-Terhal results are statements on the geometry of $\bbR^D$.
A ball of area $A$ in the Euclidean plane can be split into two equally-sized half balls by a line segment of length $\sqrt{A}$.
As a consequence, one expects a graph nicely embedded in such a ball to have a separator of size $O(\sqrt{A})$.
Similarly, a ball of area $A$ in $2$-dimensional hyperbolic space has a diameter of size $O(\log(A))$ in the limit of large balls.
We therefore expect the hyperbolic plane to perform poorly in terms of distance.

Similarly, this geometric consideration can be used to justify why the hyperbolic space might be particularly well suited to our technique.
As previously noted after Corollary \ref{cor:dimbndpoly}, Theorem \ref{thm:dimbndcmax} does not allow us to rederive the Bravyi-Poulin-Terhal bound; we cannot guarantee that the regions we create through the recursive separation have small boundaries.
However, we do not expect this to be relevant in hyperbolic space since, due to the isoperimetric inequality, the boundary of a region is proportional to its volume in the limit of large volumes. 

To formalize this correspondence between geometry and graphs, we need a precise definition of what it means for a graph to be nicely embedded in such a space.
We expect the density of vertices not to diverge, and two vertices linked by an edge should not be too far apart.
This leads to definition \ref{def:localgraph}.

Let $(M,d)$ denote a metric space $M$ equipped with a metric $d: M \times M \to \bbR_{\geq 0}$.
Denote by $\cB(x,w) = \{y \in M : d(x,y) \leq w\}$ the ball of radius $w$ around the point $x \in M$. 

\begin{definition}
	\label{def:localgraph}
	A graph $G = (V,E)$ is said to be $(\rho,w)$-local on a metric space $(M,d)$ if there exists a map $\eta : V \rightarrow M$ such that
	\begin{enumerate}
		\item $(u,v) \in E \implies d(\eta(u),\eta(v)) \leq w$.
		\item $\forall x \in M$, let $\cB^{\sharp}(x,w)  = \{v \in V : \eta(v) \in \cB(x,w)\}$ be the (possibly empty) preimage of a ball.
		Then $\forall x \in M, |\cB^{\sharp}(x, 2w)| \leq \rho$.
	\end{enumerate}
\end{definition}

A recent result by Kisfaludi-Bak \cite{kisfaludi2020hyperbolic} demonstrates that $(\rho,w)$-local graphs embedded in $\bbH^D$ have small separators.
We begin by repeating Theorem 2 from \cite{kisfaludi2020hyperbolic} as it applies to this class of graphs -- see Section \ref{subsec:clarification-kisfaludi} for details.
\begin{lemma}
	\label{thm:kisbak}
	Let $D \geq 2$, let $G$ be $(\rho,w)$-local in $\bbH^D$ .
	Then $s_G^{(D-1)/D}(r) = O(f(r))$, where
	\begin{enumerate}[label=(\roman*)]
		\item if $D=2$, then $ f(r) = O(\log(r))$, and
		\item if $D\geq 3$, then $ f(r) = O(r^{(D-2)/(D-1)})$.
	\end{enumerate}
\end{lemma}

From our previous results, we can then prove the following theorem.
\begin{theorem}
	If $\scrC = \{\cC_n\}$ is a family of $\dsl n,k,d \dsr$ LDPC codes such that the corresponding connectivity graphs $\cG = \{G_n\}$ are $(\rho,w)$-local in $\bbH^D$.
	Then we have the bounds
	\begin{enumerate}[label=(\roman*)]
		\item if $D=2$, then $d = O(\log(n))$ and $k \frac{d^2}{\log(d)^2} = O(n)$, and
		\item if $D\geq 3$, then $d = O(n^{(D-2)/(D-1)}) $ and $k d^{\frac{2}{D-1}} = O(n)$.
	\end{enumerate}
\end{theorem}
\begin{proof}
	The distance bounds are a trivial application of Theorem \ref{thm:distancebnd}.
	For the 2D case, one finds $\cS_d(r) = O(r\frac{\log(d)}{d})$.
	Hence $k = O(n \frac{\log(d)^2}{d^2})$ from Theorem \ref{lem:sepbpt}.
	In the $D \geq 3$ case, we have $k = O(d^{-\frac{2}{D-1}}n)$ from Corollary \ref{cor:dimbndpoly}.
\end{proof}

We can see from this result that the distance of the $D$-dimensional hyperbolic codes for $D \geq 3$ obeys the same upper bound as $(D-1)$-dimensional Euclidean local codes.

Note that these results do not apply to hyperbolic manifolds of the form $\bbH^D/\Gamma$ as quotienting by $\Gamma$ can change the size of the separator completely.
A straightforward consequence is that a graph on a 2-torus does not necessarily have a $O(\log(n))$ separator.
Fortunately, in the case of hyperbolic surfaces, we can still bound the separator as a function of the genus.
We turn next to these codes.

\subsection[]{Surfaces of genus $g$}
The class of 2D topological codes has generated a wealth of literature and is among the most likely candidates for physical implementation in the near future.
One could attribute this popularity to their relative ease of implementation and tractable properties.
Unfortunately, due to a result by Delfosse, these codes are known to be strongly limited and are constrained by $kd^2 = O(\log(k)^2n)$ \cite{delfosse2013tradeoffs}.
Here we generalize this bound for local codes on an arbitrary surface of genus $g$, denoted $\Sigma_g$, and we prove that for fixed $g$, $d = O(\sqrt{n})$, which is saturated by the surface code.

Topological graph theory provides a very natural bridge between graphs embeddable on a surface and their separability \cite{gilbert1984separator,djidjev1995planarization, kelner2006spectral,aleksandrov1996linear}.
We employ a result due to Dujmovi\'c, Eppstein and Wood \cite{dujmovic2015genus} which states that graphs embedded in $\Sigma_g$ with planarity $p$ have bounded separators.
A graph is said to be $p$-planar if it can be drawn with at most $p$ crossings on each edge.
\cite{dujmovic2015genus} proved that any graph that can be embedded in $\Sigma_g$ that is $p$-planar has a separator of size $O(\sqrt{(g+1)(p+1)n})$.

Observe that all $\rho$-local graphs that can be embedded on $\Sigma_g$ must be $t$-planar for some constant $t$.
This is because:
\begin{enumerate}
	\item Every edge $(a,b)$ can be contained within a ball $B$ of radius $w$.
	\item For any edge $(c,d)$ crossing $(a,b)$, $c$ and $d$ must be at a distance less than $w$ to some point $p$ in the ball, which is at a distance at most $w$ from $a$ and $b$. 
	\item Since the number of points at a distance less than $2w$ is bounded by $\rho$, there can be at most $\rho$ crossings
\end{enumerate}
Together, these observations imply that any $\rho$-local graph on $\Sigma_g$ is $t$-planar for some constant $t$.
This implies the following result.

\begin{theorem}
	\label{thm:genusbound}
	Let $\scrC = \{\cC_n\}$ be a family of $\dsl n,k,d \dsr$ LDPC codes such for every connectivity graph $G \in \cG = \{G_n\}$, G is $(\rho,w)$-local on a surface of genus $g$.
	We have the bounds
	\begin{enumerate}[label=(\roman*)]
		\item $d = O(\sqrt{g n})$, and
		\item $kd = O(g n)$.
	\end{enumerate}
\end{theorem}

We can also use this result to think about implementations of good quantum codes.
We might wish to implement codes such that there are as few edge crossings as possible.
If we consider implementing codes on a flat surface, then we would need a significant number of edges overlapping.
Indeed the number of edges crossing would have to scale as $n$.
This, in turn, would mean that the code is no longer $(\rho,w)$-local.
\section{Conclusions}
In this paper, we have shown that there is an intimate relation between quantum codes and the graphs on which they are defined.
Given a code, we can obtain the connectivity graph from which we can infer properties of the associated quantum codes.
We have three main results.
First, we found that the distance of the quantum LDPC code is bounded by the size of the separator of the associated connectivity graph.
Second, we found that the code dimension of a code is bounded as a function of the size of the separator via a recurrence relation.
Third, we found that transversal gates can only implement a limited set of transformations depending on the connectivity of the graph.
Together, the first two results state that we have good quantum LDPC codes only when the connectivity graph contains an expander.

We explored the properties of codes embedded in $D$-dimensional hyperbolic space.
In particular we found that a local code in $D$-dimensional hyperbolic space obeys $ d = O(n^{(D-2)/(D-1)})$ and that for closed $2$-manifolds with genus $g$, local codes obey $d = O(\sqrt{gn})$.

These results raise many interesting questions.

\begin{enumerate}	
	\item Non-sparse graphs can often be well approximated by sparser graphs \cite{batson2013spectral, chekuri2014degree, althfer1993sparse, abraham2008nearly}. This naturally leads to the following question: can the bound on the distance from the treewidth be made independent from the maximum degree of the connectivity graph?
	
	\item Low-connectivity codes have poor performance, but do all codes with poor connectivity have low-dimensional local embeddings? \cite{abraham2018metric, abraham2005metric, sidiropoulos2017metric, abraham2009low,matsubayashi2015separator} Note that a strict locality requirement would be hard to satisfy: consider a connectivity graph that has the form of a $\delta$-regular tree, then this graph cannot be embedded locally, as a ball in the tree can grow much quicker than a ball in $\bbR^D$ for any constant $D$.
	
	\item High connectivity is necessary for good quantum codes.
	Can it be proven to be a sufficient condition given some minimal extra assumptions? 
	Can we distinguish sufficient and insufficient connectivity through another graph metric? For example some families of expander graphs have bounded book thickness \cite{dujmovic2016layout}, but this metric is not bounded for sparse graphs \cite{bart2006bounded}.
	
	\item The connectivity graph representation relies on selecting a particular basis for our generators, but there exist more algebraic representation of a graph \cite{spielman2007spectral}. 
	Can the results we present here be generalized to be basis independent, for example using spectral partitioning? \cite{spielman1996spectral} 
	
	\item The recursive separation method we use is rather naive. Is it possible to formulate a better one? \cite{federickson1987fast, henzinger1997faster,bansal2014minmax}
	
\end{enumerate}

We add that we use the same techniques as reference \cite{bravyi2010tradeoffs} to also prove bounds on the code dimension of classical codes based on the graph separator in Appendix \ref{app:classicalbound}.
We would like to highlight that in contrast to local codes, there are many basic open questions concerning quantum LDPC codes.
For a broad discussion on the subject, we point the interested reader to a review by Breuckmann and Eberhardt \cite{breuckmann2021ldpc}.

\section{Acknowledgements}
This paper is dedicated to the memory of David Poulin, a role model as a researcher and a mentor.
David inspired and encouraged us to explore fundamental questions in quantum error correction while simultaneously studying consequences for real-world implementations.
His presence will be missed.

We would like to thank Guillaume Duclos-Cianci for facilitating this collaboration; Stefanos Kourtis and Patrick Hayden for helpful discussions; Chris Chubb, Nicolas Delfosse, Anthony Leverrier, Noah Shutty and Christophe Vuillot for catching mistakes in an earlier draft and comments that helped improve the paper; and Sam Cree for a clear presentation of the Bravyi-Koenig bound which inspired section \ref{subsec:transversal}.
AK is supported by the Bloch postdoctoral fellowship at Stanford University and grants NSF CCF-1844628 and NSF CCF-1763299.

\bibliographystyle{abbrvnat}
\bibliography{references}

\appendix

\appendix

\section{Relation between $s^\alpha_G(n)$ and $s^{2/3}_G(n)$}

\begin{lemma}
	\label{lem:salpha}
	For a graph $G$ on $n$ vertices, we have $s^{\alpha}_G(n) =  \Theta(s^{2/3}_G(n))$ for any $\alpha \in [\frac{2}{3},1)$.
\end{lemma}
\begin{proof}
	Since any $\frac{2}{3}$-separator is a $\alpha$-separator for $\alpha > 2/3$, then $s^{\alpha}_G(n)\leq s^{2/3}_G(n)$.
	
	We then only have to show that $s^{2/3}_G(n) =  O(s^{\alpha}_G(n))$. 
	We note that if, after removing a $\alpha$-separator, $A,B$ do not satisfy $|A|,|B| \leq \frac{2}{3} n$, then they both can be separated again $\lfloor t+1 \rfloor$ times such that $\alpha^t = \frac{1}{3}$.
	At this point the vertices have been partitioned into $V = \sqcup_{i=1}^{i=l} C_i \sqcup_{j=1}^{j=m} D_j$, where the $C_i$ induce mutually disjoint subgraphs with $|C_i| \leq n/3$, and the $D_j$ are a set of $\alpha$-separators.
	Note that there are at most $m \leq \sum_{u = 1}^{u = \lfloor t+1 \rfloor} 2^u =  O(2^{t})$ such $D_j$.
	We will want to show that $\cup_j D_j$ is a $\frac{2}{3}$-separator.
	
	We can now group the $C_i$ in order to form a partition $A,B$ such that $|A|,|B| \leq 2n/3$.
	Consider the set $A = \cup_{i=1}^{i=p} C_i$ such that $|A| \leq 2n/3$, but $|A| + C_{p+1} \geq 2n/3$ for some $p$.
	Since $|C_{p+1}| \leq n/3$, then $|A| \geq n/3$.
	Let $B = \cup_{i = l+1}^{i=l} C_i$, then $|B| \leq n - |A| \leq 2n/3$. We conclude that $s_G^{2/3}(n) \leq \sum_{j=1}^{j=m}|D_j| \leq m \cdot s^{\alpha}(n) =  O(2^t s^{\alpha}_G(n))$, or $s_G^{2/3}  =  O(2^{\log_\alpha(1/3)}s^\alpha_G(n))$.
	
	Therefore $s^\alpha_G(n) =  \Theta(s^{2/3}_G(n))$.
\end{proof}

\section{Closed-form expression for recurrence relation}
\label{app:closed-form-recurrence}

\begin{lemma}
	\label{lem:solrecursion}
	Consider the function $\cS_d$ defined by the recurrence relation
	\[
	\cS_d(r) = \beta r^c + \cS_d(\alpha_r r) + \cS_d((1-\alpha_r)r)~,
	\]
	where $0 < c \leq 1 + \log_n(1-\alpha)<1$,  $\frac{1}{2} \leq \alpha_r \leq \alpha < 1$, and $\beta$ is a constant.
	Furthermore, $\cS_d$ obeys $\cS_d(r) = 0$ for $r \leq d - 1 $.
	Then there is a closed-form expression for $\cS_d$,
	\[
	\cS_d(r) = O\left(\frac{r}{d^{1-c}}\log(n)\right)~.
	\]
	If $c$ is upper bounded by a constant, then
	\[
	\cS_d(r) = O\left(\frac{r}{d^{1-c}}\right)~.
	\]
\end{lemma}
\begin{proof}
	Let $d' \equiv d-1$. We will begin by showing by induction that $\cS_d(r) \leq f(c) \frac{r}{d'^{1-c}}-g(c)r^c$ for $r \geq (1-\alpha)d'$ for some functions $f,g$ of $c$.
	
	In order to prove the base case, we will consider $(1-\alpha)d' \leq r \leq d'$, and show that $f(c) \frac{r}{d'^{1-c}}-g(c)r^c \geq \cS_d(r)$ if we take $f(c) = \frac{g(c)}{(1-\alpha)^{1-c}}$. Indeed, since by assumption, $\frac{1}{d'} \geq \frac{1-\alpha}{r}$ then one can readily verify that
	
	\begin{align*}
		f(c) \frac{r}{d'^{1-c}}-g(c)r^c  &= \frac{g(c)r}{(1-\alpha)^{1-c}d'^{1-c}} -g(c)r^c\\
		&\geq \frac{g(c)r(1-\alpha)^{1-c}}{(1-\alpha)^{1-c}r^{1-c}} -g(c)r^c\\
		& \geq 0 = \cS_d(r)~.
	\end{align*}
	As the base case is verified, we can then verify the expression holds for $r > d'$.
	\begin{align*}
		\cS_d(r) &= \beta r^c + \cS_d(\alpha_r r) + \cS_d((1-\alpha_r)r) \\
		&\leq \beta r^c + f(c)\frac{\alpha_rr}{d'^{1-c}} - g(c)(\alpha_r r)^c + \nonumber \\
		&f(c)\frac{(1-\alpha_r)r}{d'^{1-c}} - g(c)((1-\alpha_r)r)^c \\
		& = f(c)\frac{r}{d'^{1-c}}-\\
		&g(c)r^c\left(-\frac{\beta}{g(c)}+\alpha_r^c+(1-\alpha_r)^c\right)~.
	\end{align*}
	The induction step is then satisfied if and only if $-\frac{\beta}{g(c)}+\alpha_r^c+(1-\alpha_r)^c \leq 1$, which is equivalent to $g(c) \geq\frac{\beta}{\alpha_r^c+(1-\alpha_r)^c-1}$ for any possible value of $\alpha_r$. 
	From here one we thus take $g(c) \equiv \frac{\beta}{\alpha^c+(1-\alpha)^c-1} \geq \frac{\beta}{\alpha_r^c+(1-\alpha_r)^c-1}$.
	Since $c<1$, this is always well defined.
	
	We have thus showed that	
	\begin{align}
		\cS_d(r) & \leq  f(c) \frac{r}{d'^{1-c}}-g(c)r^c \\
		& \leq  f(c) \frac{r}{d'^{1-c}} \\
		& =     \frac{(1-\alpha)^{c-1}}{\alpha^c+(1-\alpha)^c-1}\frac{\beta r}{d'^{1-c}} ~.
	\end{align}
	If $c$ is a constant then $\cS_d(r) =  O(\frac{r}{d'^{1-c}})$.
	Since $d' = \Omega(d)$, then $\cS_d(r) =  O(\frac{r}{d^{1-c}})$.
	
	In the case where $c$ might depend on $n$, we remind the reader that $c \leq 1 + \log_n(1-\alpha)$, which gives 
	\begin{align}
		\cS_d(r) \leq \frac{(1-\alpha)^{\log_n(1-\alpha)}}{\alpha^{1 + \log_n(1-\alpha)}+(1-\alpha)^{1 + \log_n(1-\alpha)}-1}\frac{\beta r}{d'^{1-c}}~.
	\end{align}
	First we can note that $\lim\limits_{n \rightarrow \infty} (1-\alpha)^{\log_n(1-\alpha)} = 1$.
	Secondly, we verify using Mathematica \cite{mathematica} that
	\begin{align*}
		&\lim\limits_{n \rightarrow \infty} \log(n) \left( \alpha^{1 + \log_n(1-\alpha)}+(1-\alpha)^{1 + \log_n(1-\alpha)}-1 \right)\\
		=&\log(1 - \alpha) (\alpha \log(\alpha) + (1 - \alpha) \log(1 - \alpha)) > 0~.
	\end{align*}
	Equivalently, for any $\epsilon > 0$, there exist $n_0$ such that for all $n\geq n_0$ we have
	
	\begin{align*}
		&\frac{(1-\alpha)^{\log_n(1-\alpha)}}{\alpha^{1 + \log_n(1-\alpha)}+(1-\alpha)^{1 + \log_n(1-\alpha)}-1} \\
		\leq &\frac{\log(n)(1+ \epsilon)}{\log(1 - \alpha) (\alpha \log(\alpha) + (1 - \alpha) \log(1 - \alpha))} 
	\end{align*}
	
	When $\alpha$ is a constant, we therefore conclude $\cS_d(r) =  O\left(\frac{r}{d^{1-c}}\log(n)\right)$.
	
\end{proof}

\section{Classical codes}
\label{app:classicalbound}
In \cite{bravyi2010tradeoffs}, the authors also derive a bound on the parameters of classical codes using the following lemma.

\begin{lemma}[Bravyi, Poulin, Terhal]
	\label{lem:bptclassical}
	Consider a classical code defined on bits $Q = A \sqcup B$.
	Consider $A = \sqcup A_i$ with $A_i$ correctable, and no constraint acts on two different $A_i$.
	Then $A$ is correctable, and
	\begin{equation}
		k \leq |B|.
	\end{equation}
\end{lemma}

Our results extend naturally to the classical setting.

\begin{corollary}
	For a family of classical codes $\{\cC_n\}$, let $\cS_d$ be defined as in Definition \ref{def:cS}, then $k \leq \cS_d(n)$.
\end{corollary}

\begin{proof}
	This follows from applying Lemma \ref{lem:bptclassical} to \ref{lem:recpartition} taking $B \equiv \comp{A}$.
\end{proof}

\section{Bound on codes defined by commuting projectors}
\label{app:allcodesbound}

The Bravyi-Poulin-Terhal bound \cite{bravyi2010tradeoffs} applies to a much larger class of codes than just stabilizer codes.
Mirroring their result, we consider a class of codes defined by a set $\{\Pi_a\}_a$ of commuting projectors where each projector $\Pi_a$ acts on some constant number of qubits, and every qubit can affect at most a constant number of projectors.
We refer to such a code as a low-density commuting projector code.
A stabilizer code is the special case where each projector can be expressed as $\Pi_a = \frac{1}{2}\left(1 + \ssS_a\right)$ for some stabilizer generators $\{\ssS_a\}_a$.
However, in general, a commuting-projector code need not have this specific structure.
The codespace $\cC$ is the space $\cC = \{\ket{\psi}: \Pi_a \ket{\psi} = \ket{\psi} \forall a\}$.
In this section, we prove that our main results extend to this general class of codes.

At the outset, this may seem difficult as the workhorse behind our results was the Cleaning Lemma.
However, there exist analogues of the Union and Expansion Lemmas.
For proofs, we refer the reader to Lemma 2 and Corollary 1 respectively of the paper by Bravyi, Poulin and Terhal \cite{bravyi2010tradeoffs}.

Before stating the lemmas, we note that the idea of a boundary, either exterior $\bdry_{+}$ or interior $\bdry_{-}$ carries over quite naturally.
Let $V$ be the set of qubits defining a code and $U \subseteq V$ be some subset of qubits.
The external boundary $\bdry_{+}U$ of $U$ is the set of qubits $v \in \comp{U}$ such that there exists a projector acting on $v$ and some $u \in U$.
The internal boundary $\partial_{-}U$ is the exterior boundary of $\comp{U}$.
The boundary $\bdry U$ is the union $\bdry_{+} U \union \bdry_{-} U$.
We say that two regions $U_1$ and $U_2$ are decoupled if there exist no projectors $\Pi_a$ that are supported jointly on the two regions.
The definition of the connectivity graph extends naturally: two qubits are connected by an edge if they are both in the support of a projector $\Pi$.
Therefore, two regions $U_1$ and $U_2$ are decoupled if there exist no edges between them in the connectivity graph.

\begin{lemma}[Generalized Union Lemma]
	\label{lem:unionnonstab}
	Let $U_1$, $U_2$ be any correctable regions such that $U_1$ and $U_2$ are decoupled.
	Suppose that $\bdry_+ U_1$ is also correctable, then $U_1 \cup U_2$ is correctable.
\end{lemma}
Observe the qualitative difference in this case from that of stabilizer codes: we also require that $\bdry_{+}U_1$ be correctable for the union to be correctable.
This difference will manifest in the bounds on code properties by making the bounds weaker in the case of commuting projector codes.

\begin{lemma}[Generalized Expansion Lemma]
	\label{lem:expansionnonstab}
	Let $U$ be a correctable set of qubits.
	Consider any region $B \subseteq V$ and $C \subseteq \comp{U}$, such that $\bdry U \subseteq BC$ and $BC$ is correctable.
	Then $U \cup C$ is correctable.
\end{lemma}

We are now ready to prove a bound on the distance.

\begin{theorem}
	Let $\cC$ be a low-density commuting projector code on $n$ qubits.
	Let $G = G(\cC)$ be the corresponding connectivity graph with bounded degree $\delta$.
	If $G$ has treewidth $\tw(G)$, then $d \leq 8 \delta^2 \tw(G)$. 
\end{theorem}
\begin{proof}
	
	For the sake of contradiction, assume $d > 8 \delta ^2 \tw(G)$.
	
	Consider a set of leaves $\{j_1,...,j_t\}$ sharing the same parent node $p$, and define $\cA_i = Q(j_i) \setminus \cup_k \bdry_+ Q(j_k)$, and $\cA = \union_i \cA_i$. 
	Each $\cA_i$ is then what remains of $Q(j_i)$ after we remove the qubits connected to another leaf.
	
	As noted before in the proof of Theorem \ref{lem:twdist}, for any $\cA_i$, we have $\bdry_+ \cA_i \subseteq \union_i \bdry_+ Q(j_i) \subseteq \cP_{\ext}$, and we remind the reader that $\cP_{\ext} = Q(p) \cup \bdry_+ Q(p)$.
	From the definition of treewidth, $|Q(p)| \leq \tw(G) + 1$, and $\bdry_+ Q(p) \leq  \delta |Q(p)|$.
	Then $|\cP_{\ext}| \leq \tw(G) + 1 + \delta (\tw(G) + 1)$.
	This implies $|\cP_{\ext}| \leq 4 \delta \tw(G)$.
	Since $d > 4 \delta \tw(G) \geq |\cP_{\ext}| \geq \bdry_+ \cA_i$, and the $\cA_i$ are decoupled, then $\cA$ is correctable by Lemma \ref{lem:unionnonstab}.
	
	We will now want to prove that $\cA \union \cP_{\ext}$ is correctable using the Generalized Expansion Lemma \ref{lem:expansionnonstab}.
	
	First, note that since $\bdry_+ \cA \subseteq \cP_{\ext}$, then $\bdry \cA = \bdry_- \cA \union \bdry_+ \cA \subseteq  \bdry_- \cA \union \cP_{\ext}$. By the Generalized Expansion Lemma, it only remains to prove that $\bdry_- \cA \union \cP_{\ext}$ is correctable.
	It was already noted that $\bdry_+ \cA \subseteq \cup_i \bdry_+ Q(j_i)$, and trivially $\bdry_- \cA \subseteq \bdry_+ \bdry_+ \cA$.
	We can then easily bound the size of this region:  $|\bdry_- \cA \union \cP_{\ext}| \leq |\bdry_+ \bdry_+ \cA \union \cP_{\ext}| \leq |\bdry_+ \cP_{\ext} \union \cP_{\ext}| \leq (\delta + 1)|\cP_{\ext}| \leq (\delta + 1) \cdot 4 \delta \tw(G) \leq 8 \delta^2 \tw(G) < d$.
	As $\bdry \cA \subseteq \bdry_- \cA \union \cP_{\ext}$, and $\bdry_- \cA \union \cP_{\ext}$ is correctable, then $\cA \union \cP_{\ext}$ is correctable by Lemma \ref{lem:expansionnonstab}.
	Just as for Theorem \ref{lem:twdist}, $\cup_i Q(j_i) \cup Q(p) \subset \cA \cup \cP_{\ext}$.
	
	We can then repeat the argument as in Theorem \ref{lem:twdist} and prove the entire code is correctable: a contradiction if the code is to encode at least one logical qubit
\end{proof}

We now turn to the case of the bound $k$.

\begin{theorem}
	Let $\cC$ be a commuting-projector code on $n$ qubits.
	Let $G = G(\cC)$ be the corresponding connectivity graph of bounded degree $\delta$.
	Let $S_d$ be defined as in Definition \ref{def:cS}, then we have the bound $k \leq \cS_{d/\delta} \circ \cS_{d/\delta}(n)$.	 
\end{theorem}
\begin{proof}
	We remind the reader that Lemma \ref{lem:bptabc} also holds for non stabilizer codes \cite{bravyi2010tradeoffs}.
	However due to the restriction from Lemma \ref{lem:unionnonstab} that $\bdry_+ M_1$ has to be correctable, we will have to adapt our use of the recursive separation: instead of creating regions of size $d$, we will stop at $d/\delta$.
	From \ref{lem:recpartition}, we can find $A$ such that $A$ is the union of disjoint subsets $\{V_\bullet\}$ of size less than $d/\delta$, and $|\comp{A}| \leq \cS_{d/\delta}(n)$.
	For every $V_\bullet$, $|\bdry_+ V_\bullet| \leq \delta |V_\bullet| < d$, which shows that $A$ is correctable.
	
	By applying the same argument as in Lemma $\ref{lem:sepbpt}$, we find $k \leq \cS_{d/\delta} \circ \cS_{d/\delta}(n)$.
\end{proof}
\section[]{Conditions for separators on $\bbH^{D}$}
\label{subsec:clarification-kisfaludi}
In this section, we wish to clarify when we can apply Kisfaludi-Bak's results.
Their statement is not in terms of $(\rho,w)$-local graphs, but instead in terms of a certain class called noisy uniform ball graphs (NUBG).
This class is defined as follows.
\begin{definition}[Noisy uniform ball graphs (NUBG)]
	\label{def:nubg}
	Let $(M,d)$ be a metric space.
	Let $\sigma > 0$ and $\nu \geq 1$ be fixed constants.
	A graph $G = (V,E) \in \nubg_{\bbH^D}(\sigma,\nu)$ if there is a function $\eta: V \to M$ such that for all pairs $v,w \in V$, we have
	\begin{enumerate}
		\item $d(\eta(v), \eta(w)) < 2\sigma \implies (v,w) \in E$.
		\item $d(\eta(v), \eta(w)) \geq 2\nu\sigma \implies (v,w) \not\in E$.
	\end{enumerate}
	Pairs of vertices $v,w$ where $d(\eta(v),\eta(w)) \in [2\sigma, 2\sigma\nu]$ can either be connected or disconnected.
\end{definition}

Definition \ref{def:nubg} requires all vertices close enough to another vertex to be connected.
Note that a $(\rho,w)$-local graph $G$ can be extended to a $(w/2, 2)$-NUBG graph by adding edges between any two vertices that are a distance $w/2$ away, written $G_{w/2} = (V_{w/2}, E_{w/2})$.
The vertex set $V_{w/2} = V$ and the edges in the modified graph $G_{w/2}$ obey the condition
\begin{align*}
    (u,v) \in E_{w/2} \Leftrightarrow d(\eta(u), \eta(v)) \leq w~.
\end{align*}
Since the density $\rho$ is constant, this modification adds at most a constant number of edges.
Note then that a separator $S$ for the NUBG graph $G_{w/2}$ is also a separator for $G$.

For completeness, we repeat the definition of Theorem $2$ from \cite{kisfaludi2020hyperbolic} using the language of NUBG graphs as in the original paper.
\begin{lemma}
    \label{thm:kisbaknubg}
       Let $D \geq 2$, $\sigma > 0$ and $\nu \geq 1$ be constants, and let $G$ be $\nubg_{\bbH^D}(\sigma,\nu)$ in $\bbH^D$ .
       Then,
       \begin{enumerate}[label=(\roman*)]
          \item if $D=2$, then $ s_G(r) = O(\log(r))$, and
          \item if $D\geq 3$, then $ s_G(r) = O(r^{(D-2)/(D-1)})$.
       \end{enumerate}
\end{lemma}

\end{document}